\newif\ifprocs
\procstrue
\procsfalse
\ifprocs 
\else
\documentclass[11pt,USenglish]{article}
\usepackage[margin=1.0in]{geometry}
\usepackage[T1]{fontenc}
\usepackage[utf8]{inputenc}
\usepackage[english]{babel}
\fi

\usepackage{mwe}
\usepackage{caption}
\usepackage{xcolor}
\usepackage{amsmath}
\usepackage{amsthm}
\usepackage{amssymb}
\usepackage{complexity}
\usepackage{graphicx}
\usepackage{amsfonts}
\usepackage{xspace}
\usepackage[export]{adjustbox}
\usepackage[colorinlistoftodos,prependcaption,textsize=tiny]{todonotes} 
\usepackage[noend, noline, ruled, linesnumbered]{algorithm2e}

\ifprocs
\else
\usepackage{ifpdf}
\ifpdf
\usepackage[pdftex,colorlinks,allcolors=blue]{hyperref}
\else
\usepackage[hypertex,colorlinks,linkcolor=blue,citecolor=blue,filecolor=blue,urlcolor=blue]{hyperref}
\fi
\fi

\ifprocs
\else
\newtheorem{theorem}{Theorem}[section]
\newtheorem{lemma}[theorem]{Lemma}

\newtheorem{corollary}[theorem]{Corollary}
\newtheorem{definition}[theorem]{Definition}

\newtheorem{problem}[theorem]{Problem}

\newtheorem{property}[theorem]{Property}

\fi

\theoremstyle{plain}
\newtheorem{claim}[theorem]{Claim}

\newcommand{\ProblemName}[1]{\textsf{#1}}

\newcommand{\kSTMVC}{\ProblemName{kSTMVC}\xspace}

\newcommand{\kAPMVC}{\ProblemName{kAPMVC}\xspace}
\newcommand{\APMVC}{\ProblemName{APMVC}\xspace}

\newcommand{\tO}{\ensuremath{\tilde O}}
\newcommand{\bigO}{\ensuremath{O}}

\providecommand{\card}[1]{\lvert#1\rvert}
\DeclareMathOperator{\rank}{rank}

\newcommand{\stmincut}{minimum $s$-$t$ cut\xspace}
\newcommand{\stmincuts}{minimum $s$-$t$ cuts\xspace}

\let\poly\undefined
\DeclareMathOperator{\poly}{poly}
\let\polylog\undefined
\DeclareMathOperator{\polylog}{polylog}
\let\deg\undefined
\DeclareMathOperator{\deg}{deg}
\newcommand\eps{\varepsilon}

\newcommand{\witness}{\textsc{Witness}\ \textsc{Superset}\xspace}

\newcommand{\lbar}{\overline}
\newcommand{\ignore}[1]{}
\newcommand{\reach}[2]{\ensuremath{#1 {\leadsto} #2}}
\newcommand{\nreach}[2]{\ensuremath{#1 {\not\leadsto} #2}}

\SetAlFnt{\normalsize}
\DontPrintSemicolon
\SetKwComment{Comment}{$\triangleright$\ }{}
\SetKwProg{Def}{def}{:}{}

\SetKwFunction{KwSet}{WitnessSuperset}

\graphicspath{{./figures/}}

\title{Faster Algorithms for All-Pairs Bounded Min-Cuts}

\author{
Amir Abboud\footnote{IBM Almaden Research Center, US. Email: \texttt{amir.abboud@ibm.com}}
\qquad
Loukas Georgiadis\footnote{University of Ioannina, Greece. Email: \texttt{loukas@cs.uoi.gr}}
\qquad
Giuseppe F. Italiano\footnote{LUISS University, Rome, Italy. Email: \texttt{gitaliano@luiss.it}}
\qquad
Robert Krauthgamer\footnote{Weizmann Institute of Science, Israel. Email: \texttt{robert.krauthgamer@weizmann.ac.il}.\newline Work supported in part by
  ONR Award N00014-18-1-2364,
  Israel Science Foundation grant \#1086/18,
  a Minerva Foundation grant,
  and a Google Faculty Research Award.}
\qquad\\
Nikos Parotsidis\footnote{University of Copenhagen, Denmark. Email: \texttt{nipa@di.ku.dk}}
\qquad
Ohad Trabelsi\footnote{Weizmann Institute of Science, Israel. Email: \texttt{ohad.trabelsi@weizmann.ac.il}. Partly done at IBM Almaden Research Center, US.}
\qquad
Przemys\l{}aw~Uzna\'nski\footnote{University of Wroc\l{}aw, Poland. Email: \texttt{puznanski@cs.uni.wroc.pl}}
\qquad
Daniel Wolleb-Graf\footnote{ETH Z\"urich, Switzerland. Email: \texttt{daniel.graf@inf.ethz.ch}}
}

\begin{document}

\maketitle

\begin{abstract}
The All-Pairs Min-Cut problem (aka All-Pairs Max-Flow) asks to compute
a minimum $s$-$t$ cut (or just its value) for all pairs of vertices $s,t$.
We study this problem in \emph{directed graphs} with
\emph{unit} edge/vertex capacities (corresponding to edge/vertex connectivity).
Our focus is on the \emph{$k$-bounded} case, where the algorithm has to find
all pairs with min-cut value less than $k$, and report only those.
The most basic case $k=1$ is the Transitive Closure (TC) problem,
which can be solved in graphs with $n$ vertices and $m$ edges 
in time $O(mn)$ combinatorially, and in time $O(n^{\omega})$
where $\omega<2.38$ is the matrix-multiplication exponent.
These time bounds are conjectured to be optimal. 

We present new algorithms and conditional lower bounds
that advance the frontier for larger $k$, as follows:
\begin{itemize}
\item
A randomized algorithm for \emph{vertex capacities} 
that runs in time $\bigO((nk)^{\omega})$.
This is only a factor $k^\omega$ away from the TC bound,
and nearly matches it for all $k=n^{o(1)}$. 
\item
Two deterministic algorithms for \emph{edge capacities} (which is more general)
that work in DAGs and further reports a minimum cut for each pair. 
The first algorithm is combinatorial (does not involve matrix multiplication)
and runs in time $\bigO(2^{\bigO(k^2)}\cdot mn)$. 
The second algorithm can be faster on dense DAGs and runs in time $\bigO((k\log n)^{4^k+o(k)}\cdot n^{\omega})$. 
Previously, Georgiadis et al.\ [ICALP 2017], could match the TC bound
(up to $n^{o(1)}$ factors) only when $k=2$, and now our two algorithms 
match it for all $k=o(\sqrt{\log n})$ and $k=o(\log\log n)$.
\item
The first super-cubic lower bound of $n^{\omega-1-o(1)} k^2$ time
under the $4$-Clique conjecture,
which holds even in the simplest case of DAGs with unit vertex capacities.
It improves on the previous (SETH-based) lower bounds
even in the unbounded setting $k=n$. 
For combinatorial algorithms, our reduction implies an $n^{2-o(1)} k^2$ conditional lower bound.
Thus, we identify new settings where the complexity of the problem is (conditionally) higher than that of TC.
\end{itemize}

Our three sets of results are obtained via different techniques. 
The first one adapts the network coding method of Cheung, Lau, and Leung [SICOMP 2013] to vertex-capacitated digraphs. 
The second set exploits new insights on the structure of latest cuts
together with suitable algebraic tools. 
The lower bounds arise from a novel reduction of a different structure
than the SETH-based constructions. 

\end{abstract}

\section{Introduction}
Connectivity-related problems are some of the most well-studied problems in graph theory and algorithms, and have been thoroughly investigated in the literature.
Given a directed graph $G=(V,E)$ with $n=|V|$ vertices and $m=|E|$ edges,%
\footnote{We sometimes use arcs when referring to directed edges, 
  or use nodes instead of vertices.} 
perhaps the most fundamental such problem is to compute a \emph{\stmincut}, i.e., a set of edges $E'$ of minimum-cardinality such that $t$ is not reachable from $s$ in $G \setminus E'$. 
This \stmincut problem is well-known to be equivalent to maximum $s$-$t$ flow, as they have the exact same value \cite{ford1962flows}. 
Currently, the fastest algorithms for this problem
run in time $\tO(m \sqrt{n} \log^{\bigO(1)} U)$ \cite{Lee2014Path}
and $\tO(m^{10/7} U^{1/7})$ (faster for sparse graphs) \cite{Madry2016Computing},
where $U$ is the maximum edge capacity (aka weight).%
\footnote{The notation $\tO(\cdot)$ hides polylogarithmic factors.}

The central problem of study in this paper is All-Pairs Min-Cut
(also known as All-Pairs Max-Flow),
where the input is a digraph $G=(V,E)$
and the goal is to compute the \stmincut value for all $s,t\in V$. 
All our graphs will have \emph{unit} edge/vertex capacities (aka uncapacitated),
in which case the value of the \stmincut is just
the maximum number of disjoint paths from $s$ to $t$
(aka edge/vertex connectivity), by~\cite{menger}. 
We will consider a few variants:
vertex capacities vs.\ edge capacities,%
\footnote{The folklore reduction where each vertex $v$ is replaced by two
  vertices connected by an edge $v_{in} \to v_{out}$ shows that
  in all our problems,
  vertex capacities are no harder (and perhaps easier) than edge capacities.
  Notice that this is only true for directed graphs. 
}
reporting only the value vs.\ the cut itself (a witness), 
or a general digraph vs.\ a directed acyclic graph (DAG).
For all these variants, we will be interested in the \emph{$k$-bounded} version
(aka \emph{bounded min-cuts}, hence the title of the paper)
where the algorithm needs to find which \stmincuts
have value less than a given parameter $k<n$, and report only those.  
Put differently, the goal is to compute, for every $s,t\in V$, 
the minimum between $k$ and the actual \stmincut value.
Nonetheless, some of our results (the lower bounds) are of interest even without this restriction.

The time complexity of these problems should be compared against the fundamental special case that lies at their core ---
the Transitive Closure problem (aka All-Pairs Reachability),
which is known to be time-equivalent to Boolean Matrix Multiplication,
and in some sense, to Triangle Detection \cite{VW18}.
This is the case $k=1$, and it can be solved in time
$O(\min\{ mn , n^\omega \})$, where $\omega<2.38$ is the matrix-multiplication exponent~\cite{matrix_mult:cw,LeGall:2014,Williams:2012}; the latter term is asymptotically better for dense graphs, but it is not \emph{combinatorial}.%
\footnote{Combinatorial is an informal term to describe algorithms
  that do not rely on fast matrix-multiplication algorithms,
  which are infamous for being impractical.
  See~\cite{Abboud2014,ABV15b_parsing} for further discussions.
}
This time bound is conjectured to be optimal for Transitive Closure, 
which can be viewed as a conditional lower bound for All-Pairs Min-Cut;
but can we achieve this time bound algorithmically,
or is All-Pairs Min-Cut a harder problem?

The naive strategy for solving All-Pairs Min-Cut is to execute a \stmincut 
algorithm $O(n^2)$ times, 
with total running time $\tO(n^2 m^{10/7})$ \cite{Madry2016Computing}
or $\tO(n^{2.5}m)$ \cite{Lee2014Path}.
For not-too-dense graphs, there is
a faster randomized algorithm of Cheung, Lau, and Leung \cite{CheungLL13}
that runs in time $\bigO(m^\omega)$. 
For smaller $k$, some better bounds are known.
First, observe that a \stmincut can be found via $k$ iterations of the Ford-Fulkerson algorithm \cite{ford1962flows} in time $\bigO(km)$,
which gives a total bound of $\bigO(n^2mk)$. 
Another randomized algorithm of~\cite{CheungLL13}
runs in better time $O(mnk^{\omega-1})$ but it works only in DAGs. 
Notice that the latter bound matches the running time of Transitive Closure
if the graphs are sparse enough. 
For the case $k=2$, Georgiadis et al.~\cite{icalp2017GGIPU} achieved the same running time as Transitive Closure up to sub-polynomial factor $n^{o(1)}$
in all settings,
by devising two deterministic algorithms,
whose running times are $\tO(mn)$ and $\tO(n^\omega)$.

Other than the lower bound from Transitive Closure,
the main previously known result is from \cite{KrauthgamerT18},
which showed that under the Strong Exponential Time Hypothesis (SETH),%
\footnote{These lower bounds hold even under the weaker assumption
  that the $3$-Orthogonal Vectors problem requires $n^{3-o(1)}$ time.
}
All-Pairs Min-Cut requires, up to sub-polynomial factors, time $\Omega(mn)$
in uncapacitated digraphs of any edge density,
and even in the simpler case of (unit) vertex capacities and of DAGs.
As a function of $k$ their lower bound becomes $\Omega(n^{2-o(1)}k)$ \cite{KrauthgamerT18}. 
Combining the two, we have a conditional lower bound of
$(n^2k+n^{\omega})^{1-o(1)}$.

\subparagraph{Related Work.}
There are many other results related to our problem, let us mention a few. 
Other than DAGs, the problem has also been considered in the special cases of planar digraphs \cite{Arikati1998All,Lacki2012Single}, sparse digraphs and digraphs with bounded treewidth \cite{Arikati1998All}.

In \emph{undirected} graphs, the problem was studied extensively 
following the seminal work of Gomory and Hu \cite{Gomory1961Multi} in 1961,
which introduced a representation of All-Pairs Min-Cuts via a weighted tree,
commonly called a Gomory-Hu tree, 
and further showed how to compute it using $n-1$ executions of maximum $s$-$t$ flow.
Bhalgat et al.~\cite{BHKP07} designed an algorithm that computes a Gomory-Hu tree in uncapacitated undirected graphs in $\tO(mn)$ time,
and this upper bound was recently improved \cite{AKT}. 
The case of bounded min-cuts (small $k$) in undirected graphs was studied by Hariharan et al.~\cite{Hariharan2007Efficient}, motivated in part by applications in practical scenarios.
The fastest running time for this problem is $\tO(mk)$~\cite{Panigrahi16},
achieved by combining results 
from~\cite{Hariharan2007Efficient} and~\cite{BHKP07}.
On the negative side, there is an $n^{3-o(1)}$ lower bound for All-Pairs Min-Cut in sparse \emph{capacitated} digraphs \cite{KrauthgamerT18}, and very recently, a similar lower bound was shown for \emph{undirected} graphs with vertex capacities \cite{AKT}.

\subsection{Our Contribution}

The goal of this work is to reduce the gaps in our understanding of the All-Pairs Min-Cut problem
(see Table~\ref{tab:results} for a list of known and new results).
In particular, we are motivated by three high-level questions. 
First, how large can $k$ be
while keeping the time complexity the same as Transitive Closure?
Second, could the problem be solved in cubic time (or faster) in all settings? Currently no $\Omega(n^{3+\eps})$ lower bound is known even
in the hardest settings of the problem (capacitated, dense, general graphs).
And third, can the actual cuts (witnesses) be reported
in the same amount of time it takes to only report their values?
Some of the previous techniques, such as those of \cite{CheungLL13}, cannot do that.

\subparagraph{New Algorithms.}
Our first result is a randomized algorithm 
that solves the $k$-bounded version of All-Pairs Min-Cut 
in a digraph with \emph{unit vertex capacities} 
in time $\bigO((nk)^{\omega})$.
This upper bound is only a factor $k^\omega$ away from that of Transitive Closure, and thus matches it up to polynomial factors for any $k=n^{o(1)}$.
Moreover, any $\poly(n)$-factor improvement over our upper bound 
would imply a breakthrough for Transitive Closure (and many other problems). 
Our algorithm builds on the network-coding method of~\cite{CheungLL13},
and in effect adapts this method to the easier setting of vertex capacities,
to achieve a better running time than what is known for unit edge capacities.
This algorithm is actually more general: 
Given a digraph $G=(V,E)$ with unit vertex capacities,
two subsets $S, T \subseteq V$ and $k>0$,
it computes for all $s \in S, t \in T$ the \stmincut value
if this value is less than $k$, all in time  $\bigO((n+(\card{S}+\card{T})k)^{\omega}+\card{S}\card{T}k^{\omega})$.
We overview these results in Section~\ref{sec:general_overview},
with full details in Section~\ref{app:general}.

Three weaknesses of this algorithm and the ones by Cheung et al.~\cite{CheungLL13} are that they do not return the actual cuts, 
they are randomized, and they are not combinatorial. 
Our next set of algorithmic results deals with these issues.
More specifically, we present two deterministic algorithms for DAGs with unit edge (or vertex) capacities that compute, for every $s,t\in V$,
an actual \stmincut if its value is less than $k$.
The first algorithm is \emph{combinatorial} (i.e., it does not involve matrix multiplication) and runs in time $\bigO(2^{\bigO(k^2)}\cdot mn)$.
The second algorithm can be faster on dense DAGs and runs in time $\bigO((k\log n)^{4^k+o(k)}\cdot n^{\omega})$.
These algorithms extend the results of Georgiadis et al.~\cite{icalp2017GGIPU},
which matched the running time of Transitive Closure up to $n^{o(1)}$ factors, from just $k=2$ to any $k=o(\sqrt{\log n})$ (in the first case) and $k=o(\log\log n)$ (in the second case).
We give an overview of these algorithms in Section~\ref{sec:latest},
and the formal results are Theorems~\ref{th:dynprog} and~\ref{thm:latestDense}.

\subparagraph{New Lower Bounds.}
Finally, we present conditional lower bounds for our problem, 
the $k$-bounded version of All-Pairs Min-Cut.
As a result,
we identify new settings where the problem is harder than Transitive Closure,
and provide the first evidence that the problem cannot be solved in cubic time.
Technically, the main novelty here is a reduction from the $4$-Clique problem.
It implies lower bounds that
apply to the basic setting of DAGs with unit vertex capacities, and therefore immediately apply also to more general settings,
such as edge capacities, capacitated inputs, and general digraphs,
and they in fact improve over previous lower bounds \cite{AVY18,KrauthgamerT18} in all these settings.%
\footnote{It is unclear if our new reduction can be combined with the ideas in \cite{AKT} to improve the lower bounds in the seemingly easier case of undirected graphs with vertex capacities.} 
We prove the following theorem in Section~\ref{sec:CLB}. 

\begin{theorem}
\label{thm:CLB}
If for some fixed $\eps>0$ and any $k \in [n^{1/2},n]$,
the $k$-bounded version of All-Pairs Min-Cut can be solved
on DAGs with unit vertex capacities
in time $O( (n^{\omega-1}k^2)^{1-\eps} )$,
then $4$-Clique can be solved in time $O(n^{\omega+1 - \delta})$
for some $\delta=\delta(\eps)>0$.

Moreover, if for some fixed $\eps>0$ and any $k \in [n^{1/2},n]$
that version of All-Pairs Min-Cut can be solved combinatorially 
in time $O( (n^{2}k^2)^{1-\eps} )$,
then $4$-Clique can be solved combinatorially 
in time $O(n^{4- \delta})$ 
for some $\delta=\delta(\eps)>0$.
\end{theorem}

To appreciate the new bounds, consider first the case $k=n$,
which is equivalent to not restricting $k$. 
The previous lower bound, under SETH, is $n^{3-o(1)}$
and ours is larger by a factor of $n^{\omega-2}$.
For combinatorial algorithms, our lower bound is $n^{4-o(1)}$,
which is essentially the largest possible lower bound one can prove
without a major breakthrough in fine-grained complexity.
This is because the naive algorithm for All-Pairs Min-Cuts
is to invoke an algorithm for Max-Flow $O(n^2)$ times,
hence a lower bound larger than $\Omega(n^4)$ for our problem
would imply the first non-trivial lower bound for \stmincut. 
The latter is perhaps the biggest open question in fine-grained complexity,
and in fact many experts believe that near-linear time algorithms for \stmincut do exist, and can even be considered ``combinatorial'' in the sense that they do not involve the infamous inefficiencies of fast matrix multiplication.
If such algorithms for \stmincut do exist, then our lower bound is tight.

Our lower bound shows that as $k$ exceeds $n^{1/2 -o(1)}$,
the time complexity of $k$-bounded of All-Pairs Min-Cut
exceeds that of Transitive Closure by polynomial factors. 
The lower bound is super-cubic whenever $k\geq n^{2-\omega/2 + \eps}$.

\begin{table}[htb!]
\renewcommand{\arraystretch}{1.1}%
\begin{center}
\begin{tabular}{lllll}
\hline\hline
Time & & Input & Output & Reference\\
\hline
  $O(mn),\tO(n^{\omega})$  & deterministic & digraphs & cuts, only $k=2$ & \cite{icalp2017GGIPU} \\
  $O(n^2mk)$  & deterministic & digraphs & cuts & \cite{ford1962flows} \\
  $O(m^{\omega})$  & randomized & digraphs & cut values & \cite{CheungLL13}\\
  $O(mnk^{\omega-1})$  & randomized & digraphs & cut values & \cite{CheungLL13}\\
  $O((nk)^\omega)$ & randomized, vertex capacities  & digraphs & cut values & Theorem~\ref{Thm:Vertices}\\
  $2^{O(k^2)}mn$  & deterministic & DAGs & cuts & Theorem~\ref{th:dynprog}  \\
  $(k\log n)^{4^k+o(k)}\cdot n^{\omega}$  & deterministic & DAGs & cuts & Theorem~\ref{thm:latestDense}\\
  \hline
  $(mn+n^{\omega})^{1-o(1)}$ & based on Transitive Closure & DAGs & cut values & \\
  $n^{2-o(1)}k$ & based on SETH & DAGs & cut values & \cite{KrauthgamerT18}  \\
  $n^{\omega-1-o(1))} k^2$ & based on 4-Clique & DAGs & cut values & Theorem~\ref{thm:CLB} \\
\hline
\end{tabular}
\caption{Summary of new and known results. Unless mentioned otherwise, all upper and lower bounds hold both for unit edge capacities and for unit vertex capacitities. \label{tab:results} }
\end{center}
\end{table}

\section{Preliminaries}

We start with some terminology and well-known results on graphs and cuts. Next we will briefly introduce the main algebraic tools that will be used throughout the paper. 
We note that although we are interested in solving the $k$-bounded All-Pairs Min-Cut problem, where we wish to find  the all-pairs min-cuts of size at most $k-1$, for the sake of using simpler notation we compute the min-cuts of size at most $k$ (instead of less than $k$) solving this way the ($k+1$)-bounded All-Pairs Min-Cut problem. 

\subparagraph{Directed graphs.}
The input of our problem consists of an integer $k \geq 1$ and a \emph{directed graph}, digraph for short, $G = (V,A)$ with $n := |V|$ \emph{vertices} and $m := |A|$ \emph{arcs}.
Every arc $a = (u,v) \in A$ consists of a tail $u \in V$ and a head $v \in V$.
By $G[S]$, we denote the subgraph of $G$ \emph{induced} by the set of vertices $S$, formally $G[S] = (S, A \cap (S \times S))$.
By $N^+(v)$, we denote the \emph{out-neighborhood} of $v$ consisting of all the heads of the arcs leaving $v$. We denote by $\text{outdeg}(v)$ the number of outgoing arcs from $v$.
All our results extend to multi-digraphs, where each pair of vertices can be connected with multiple (\emph{parallel}) arcs.
For parallel arcs, we always refer to each arc individually, as if  each arc had a unique identifier.
So whenever we refer to a set of arcs, we refer to the set of their unique identifiers, i.e., without collapsing parallel arcs, like in a multi-set.

\subparagraph{Flows and cuts.}
We follow the notation used by Ford and Fulkerson~\cite{ford1962flows}.
Let $G=(V,A)$ be a digraph, where each arc $a$ has a nonnegative capacity $c(a)$.
For a pair of vertices $s$ and $t$, an $s$-$t$ flow of $G$ is a function $f$ on $A$ such that $0  \le f(a) \le c(a)$, and for every vertex $v \not= s, t$ the incoming flow
is equal to outgoing flow, i.e., $\sum_{(u,v) \in A} f(u,v) = \sum_{(v,u) \in A} f(v,u)$.
If $G$ has vertex capacities as well, then $f$ must also satisfy $\sum_{(u,v) \in A} f(u,v) \le c(v)$ for every $v \not= s, t$, where $c(v)$ is the capacity of $v$.
The value of the flow is defined as $|f|=\sum_{(s,v) \in A} f(s,v)$.  
We denote the existence of a path from $s$ to $t$ by $\reach{s}{t}$ and by $\nreach{s}{t}$ the lack of such a path.
Any set $M \subseteq A$ is an \emph{$s$-$t$-cut} if $\nreach{s}{t}$ in $G \setminus M$.
$M$ is a \emph{minimal} $s$-$t$-cut if no proper subset of $M$ is $s$-$t$-cut.
For an $s$-$t$-cut $M$, we say that its \emph{source side} is $S_M = \{x \mid \reach{s}{x} \text{ in } G\setminus M\}$ and its \emph{target side} is $T_M = \{x \mid \reach{x}{t} \text{ in } G\setminus M\}$.
We also refer to the source side and the target side as \emph{$s$-reachable} and \emph{$t$-reaching}, respectively.
An $s$-$t$ $k$-cut is a minimal cut of size $k$.
A set $\mathcal{M}$ of $s$-$t$ cuts of size at most $k$ is called a set of $s$-$t$ $\leq$ $k$-cuts. We can define vertex cuts analogously.

\begin{figure}[h]
	\begin{minipage}{\textwidth}
		\centering\includegraphics[scale=0.45]{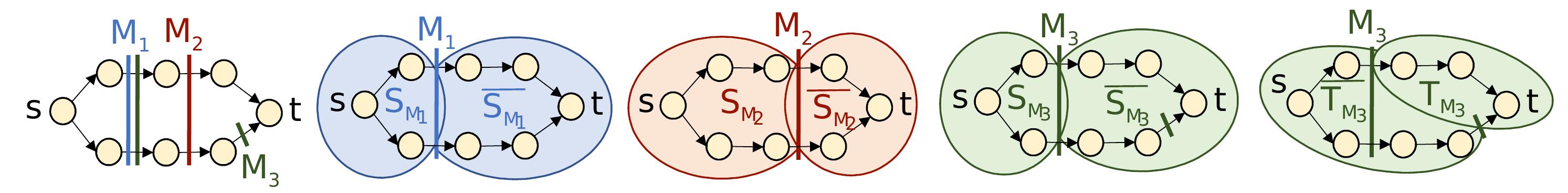}
	\end{minipage}
	\caption{A digraph with three $s$-$t$-cuts $M_1$, $M_2$, $M_3$.
		While $M_1$ and $M_2$ are minimal, 
		$M_3$ is not.
		Hence, the source side and target side differ only for $M_3$.
		This illustrates that the earlier and later orders might not be symmetric for non-minimal cuts.
		We have $M_3 < M_2$ yet $M_2 \ngtr M_3$ (and also $M_3 \leq M_2$ yet $M_2 \ngeq M_3$).
		Additionally, $M_1 \nless M_3$ yet $M_3 > M_1$ (yet both $M_1 \leq M_3$ and $M_3 \geq M_1$).
	}
	\label{fig:simple_cuts_example}
\end{figure}

\begin{figure}[h]
	\begin{minipage}{\textwidth}
		\centering\includegraphics[scale=0.6]{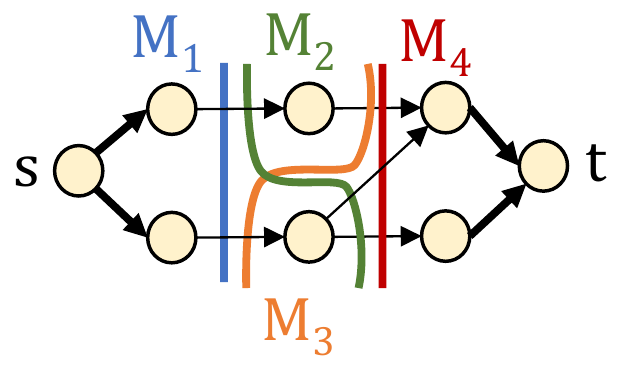}
	\end{minipage}
	\caption{A digraph with several $s$-$t$ cuts. Bold arcs represent parallel arcs which are too expensive to cut.
		$M_1$ is the earliest $s$-$t$ min-cut and $M_3$ is the latest $s$-$t$ min-cut. $M_2$ is later than $M_1$, but $M_2$ is not $s$-$t$-latest, as $M_4$ is later and not larger than $M_2$.
	}
	\label{fig:latest_example}
\end{figure}

\subparagraph{Order of cuts.}
An $s$-$t$ cut $M$ is \emph{later} (respectively \emph{earlier}) than an $s$-$t$ cut $M'$ if and only if $T_M \subseteq T_{M'}$ (resp. $S_M \subseteq S_{M'}$), and we denote it $M \ge M'$ (resp. $M \le M'$). Note that those relations are not necessarily complementary if the cuts are not minimal (see Figure~\ref{fig:simple_cuts_example} for an example).
We make these inequalities strict (i.e., `$>$' or `$<$') whenever the inclusions are proper.
We compare a cut $M$ and an arc $a$ by defining $a > M$ whenever both endpoints of $a$ are in $T_M$. Additionally, $a \geq M$ includes the case where $a \in M$.
Definitions of the relations `$\le$' and `$<$' follow by symmetry.
We refer to Figure~\ref{fig:latest_example} for illustrations.
This partial order of cuts also allows us to define cuts that are extremal with respect to all other $s$-$t$ cuts in the following sense:

\begin{definition}[$s$-$t$-latest cuts {\cite{marx2006parameterized}}]
	An $s$-$t$ cut is \emph{$s$-$t$-latest} (resp. \emph{$s$-$t$-earliest}) if and only if there is no later (resp. earlier) $s$-$t$ cut of smaller or equal size.
\end{definition}

Informally speaking, a cut is $s$-$t$-latest if we would have to cut through more arcs whenever we would like to cut off fewer vertices. This naturally extends the definition of an $s$-$t$-latest \emph{min}-cut as used by Ford and Fulkerson~\cite[Section 5]{ford1962flows}.
The notion of latest cuts has first been introduced by Marx~\cite{marx2006parameterized} (under the name of \emph{important} cuts) in the context of fixed-parameter tractable algorithms for multi(way) cut problems.
Since we need both earliest and latest cuts, we do not refer to latest cuts as important cuts.
Additionally, we use the term $s$-$t$-\emph{extremal} cuts to refer to the union of $s$-$t$-earliest and $s$-$t$-latest cuts.
\ignore{To avoid repetitions below, we state some results only for \emph{latest} cuts. However, all of them naturally extend to earliest cuts.

%
\begin{lemma}[Latest $s$-$t$ min-cut {\cite[Theorem 5.5]{ford1962flows}}]
	\label{lem:mincutinclusion}
	For any directed graph $G = (V,A)$, any maximum $s$-$t$ flow $f$ defines the same set of $t$-reaching vertices
	$T_{s,t}$ and thus defines an $s$-$t$ cut $M=A \cap (\lbar{T_{s,t}} \times T_{s,t})$, with $T_{s,t} = \{x \in V \mid \exists\text{ $x$-$t$ path in residual graph of $G$ under flow $f$}\}$.
	For any $s$-$t$ min-cut $M'$, we have $M' \le M$.
\end{lemma}
Maximum flows are not necessarily unique, but Lemma~\ref{lem:mincutinclusion} shows that the $t$-reaching cut $M$ is.

\begin{corollary}
	\label{lem:latestmincut}
	For any digraph $G$ and vertices $s$ and $t$, the latest $s$-$t$ min-cut is unique.
\end{corollary}
}

We will now briefly recap the framework of Cheung et al.~\cite{CheungLL13} as we will modify them later for our purposes.

\section{Overview of Our Algorithmic Approach}

\subsection{Randomized Algorithms on General Graphs}
\label{sec:general_overview}

In the framework of~\cite{CheungLL13} edges are encoded as vectors, so that the vector of each edge $e=(u,v)$ is a randomized linear combination of the vectors correspond to edges incoming to $u$, the source of $e$. One can compute all these vectors for the whole graph, simultaneously, using some matrix manipulations. The bottleneck is that one has to invert a certain $m \times m$ matrix with an entry for each pair of edges. Just reading the matrix that is output by the inversion requires $\Omega(m^2)$ time, since most entries in the inverted matrix are expected to be nonzero even if the graph is sparse. 

To overcome this barrier, while using the same framework, we define the encoding vectors on the nodes rather than the edges. We show that this is sufficient for the vertex-capacitated setting.
Then, instead of inverting a large matrix, we need to compute the rank of certain submatrices which becomes the new bottleneck.
When $k$ is small enough, this turns out to lead to a significant speed up compared to the running time in~\cite{CheungLL13}.

\subsection{Deterministic Algorithms with Witnesses on DAGs}
\label{sec:latest}


Here we deal with the problem of computing certificates for the $k$-bounded All-Pairs Min-Cut problem. 
Our contribution here is twofold.
We first prove some properties of the structure of the $s$-$t$-latest $k$-cuts and of the $s$-$t$-latest ${\leq}k$-cuts, which might be of independent interest. This gives us some crucial insights on the structure of the cuts, and allows us to develop an algorithmic framework which is used to solve the k-bounded All-Pairs Min-Cut problem. 
As a second contribution, we exploit our new algorithmic framework in two different ways, leading to two new algorithms which run in $\bigO(mn^{1+o(1)})$ time for $k=o(\sqrt{\log n})$ and in $\bigO(n^{\omega +o(1)})$ time for $k=o(\log \log n)$.

\begin{figure}[t]
	\begin{minipage}{\textwidth}
		\centering\includegraphics[scale=0.6]{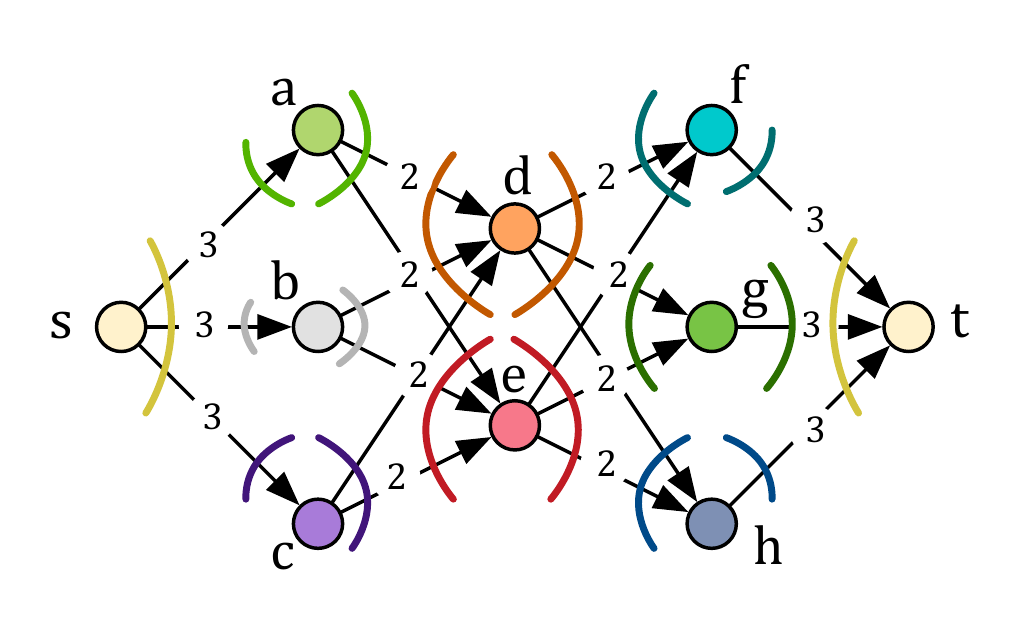}
	\end{minipage}
	\caption{A digraph where each arc appears in at least one $s$-$v$ or one $v$-$t$ min-cut.
		The numbers on the arcs denote the number of parallel arcs.
		Note that neither of the two $s$-$t$ min-cuts of size $9$ (marked in yellow) are contained within the union of any two $s$-$v$ or $v$-$t$ min-cuts. Thus, finding all those min-cuts and trying to combine them in pairs in a divide-and-conquer-style approach is not sufficient to find an $s$-$t$ min-cut.
		\label{fig:min_cut_only_counterexample}
	}
\end{figure}

Let $G=(V,A)$ be a DAG. Consider some arbitrary pair of vertices $s$ and $t$, and any $s$-$t$-cut $M$.
For every intermediate vertex $v$, $M$ must be either a $s$-$v$-cut, or a $v$-$t$-cut.
The knowledge of all $s$-$v$ and all $v$-$t$ min-cuts does not allow us to convey enough information for computing an $s$-$t$ min-cut of size at most $k$ quickly, as illustrated in Figure~\ref{fig:min_cut_only_counterexample}.
However, we are able to compute an $s$-$t$ min-cut by processing all the $s$-$v$-\emph{earliest} cuts and all the $v$-$t$-\emph{latest} cuts, of size at most $k$.
We build our approach around this insight.
We note that the characterization that we develop is particularly useful, as
it has been shown that
the number of all earliest/latest $u$-$v$ $\le$$k$-cuts can be upper bounded by $2^{\bigO(k)}$, independently of the size of the graph.

For a more precise formulation on how to recover a min-cut (or extremal $\le$$k$-cuts) from cuts to and from intermediate vertices, consider the following.
Let $A_1,A_2$ be an \emph{arc split}, that is a partition of the arc
set $A$ with the property that any path in $G$ consists of a (possibly empty) sequence of arcs from $A_1$ followed by a (possibly empty)
sequence of arcs from $A_2$
(see Definition~\ref{def:arc-split}).
Assume that for each vertex $v$ we know all the $s$-$v$-earliest ${\leq}k$-cuts in $G_1=(V,A_1)$ and all the $v$-$t$-latest ${\leq}k$-cuts in $G_2=(V,A_2)$.
We show that a set of arcs $M$ that contains as a subset one $s$-$v$-earliest ${\leq}k$-cut in $G_1$, or one $v$-$t$-latest ${\leq}k$-cut in $G_2$ for every $v$,
is a $s$-$t$-cut.
Moreover, we show that all the $s$-$t$-cuts of arcs with the above property include all the $s$-$t$-latest ${\leq}k$-cuts.
Hence, in order to identify all $s$-$t$-latest ${\leq}k$ cuts, it is sufficient to identify all sets $M$ with that property.
We next describe how we use these structural properties to compute all $s$-$t$-extremal ${\leq}k$-cuts.

We formulate the following combinatorial problem over families of sets, which is independent of graphs and cuts, that we can use to compute all $s$-$t$-extremal ${\leq}k$-cuts. 
The input to our problem is $c$ families of sets $\mathcal{F}_1, \mathcal{F}_2, \dots, \mathcal{F}_c$, where each family $\mathcal{F}_i$ consists of at most $K$ sets, and each set $F\in \mathcal{F}_i$ contains at most $k$ elements from a universe $U$.
The goal is to compute all minimal subsets $F^* \subset U, |F^*|\leq k$, for which there exists a set $F\in \mathcal{F}_i$ such that $F\subseteq F^*$, for all $1\leq i \leq c$.
We refer to this problem as \witness.
To create an instance $(s,t,A_1,A_2)$ of the \witness problem, we set $c=|V|$ and $\mathcal{F}_v$ to be all $s$-$v$-earliest ${\leq}k$-cuts in $G_1$ and all $v$-$t$-latest ${\leq}k$-cuts in $G_2$. 
Informally speaking, the solution to the instance  $(s,t,A_1,A_2)$ of the \witness problem  picks all sets of arcs that cover at least one earliest or one latest cut for every vertex.
In a post-processing step, we filter the solution to the \witness problem on the instance $(s,t,A_1,A_2)$ in order to extract all the $s$-$t$-latest ${\leq}k$-cuts.
We follow an analogous process to compute all the $s$-$t$-earliest ${\leq}k$-cuts.

\subparagraph{Algorithmic framework.}
We next define a common algorithmic framework for solving the k-bounded All-Pairs Min-Cut problem, as follows.
We pick a partition of the vertices $V_1,V_2$, such that there is no arc in $V_2 \times V_1$.
Such a partition can be trivially computed from a topological order of the input DAG.
Let $A_1,A_2,A_{1,2}$ be the sets of arcs in $G[V_1]$, in $G[V_2]$, and in $A_{1,2}=A\cap (V_1 \times V_2)$.
\begin{itemize}
	\item First, we recursively solve the problem in $G[V_1]$ and in $G[V_2]$. The recursion returns without doing any work whenever the graph is a singleton vertex.
	
	\item Second, for each pair of vertices $(s,t)$, such that $s\in V_1$ has an outgoing arc from $A_{1,2}$ and $t\in V_2$, we solve the instance $(s,t,A_{1,2},A_2)$ of \witness.
	Notice that the only non-empty earliest cuts in $(V,A_{1,2})$ for the pair $(x,y)$ are the arcs $(x,y)\in A_{1,2}$.
	
	\item Finally, for each pair of vertices $(s,t)$, such that $s\in V_1,t\in V_2$, we solve the instance $(s,t,A_1,A_{1,2}\cup A_2)$ of $\witness$.
\end{itemize}

The \witness problem can be solved naively as follows.
Let $\mathcal{F}_v$ be the set of all $s$-$v$-earliest ${\leq}k$-cuts and all $v$-$t$-latest ${\leq}k$-cuts.
Assume we have $\mathcal{F}_{v_1},\mathcal{F}_{v_2}, \dots, \mathcal{F}_{v_c}$, for all vertices $v_1,v_2,\dots,v_c$ that are both reachable from $s$ in $(V,A_{1,2})$ and that reach $t$ in $(V,A_2)$.
Each of these sets contains $2^{\bigO(k)}$ cuts.
We can identify all sets $M$ of arcs that contain at least one cut from each $\mathcal{F}_{i}$, in time $\bigO(k\cdot ({2^{\bigO(k)}})^c)$.
This yields an algorithm with  super-polynomial running time. However, we speed up this naive procedure by applying some judicious pruning, achieving a better running time of $\bigO(c \cdot 2^{\bigO(k^2)} \cdot \textrm{poly}(k))$, which is polynomial for $k = o(\sqrt{\log n})$.
In the following, we sketch the two algorithms that we develop for solving efficiently the k-bounded All-Pairs Min-Cut problem.

\subparagraph{Iterative division.}
For the first algorithm, we process the vertices in reverse topological order. When processing a vertex $v$, we define $V_1=\{v\}$ and $V_2$ to be the set of vertices that appear after $v$ in the topological order.
Notice that $V_1$ has a trivial structure, and we already know all $s$-$t$-latest ${\leq}k$-cuts in $G[V_2]$.
In this case, we present an algorithm for solving the instance $(v,t,A_{1,2},A_2)$ of the \witness problem in time $\bigO(2^{\bigO(k^2)}\cdot c \cdot \textrm{poly}(k))$, where $c=|A_{1,2}|$ is the number of arcs leaving $v$.
We invoke this algorithm for each $v$-$w$ pair such that $w\in V_2$.
For $k=o(\sqrt{\log n})$ this gives an algorithm that runs in time $\bigO(\text{outdeg}(v) \cdot n^{1+o(1)})$ for processing $v$, and $\bigO(mn^{1+o(1)})$ in total.

\subparagraph{Recursive division.}
For the second algorithm, we recursively partition the set of vertices evenly into sets $V_1$ and $V_2$ at each level of the recursion.
We first recursively solve the problem in $G[V_1]$ and in $G[V_2]$.
Second, we solve the instances $(s,t,A_{1,2},A_2)$ and $(s,t,A_1, A_{1,2}\cup A_2)$ of \witness for all pairs of vertices from $V_1 \times V_2$.
Notice that the number of vertices that are both reachable from $s$ in $(V,A_1)$ and reach $t$ in $(V,A_{1,2}\cup A_2)$ can be as high as $\bigO(n)$.
This implies that even constructing all $\Theta(n^2)$ instances of the \witness problem, for all $s,t$, takes $\Omega(n^3)$ time.
To overcome this barrier, we take advantage of the power of fast matrix multiplications by applying it into suitably defined matrices of binary codes (codewords).
At a very high-level, this approach was used by
Fischer and Meyer~\cite{fischer1971boolean} in their 
$\bigO(n^\omega)$ time algorithm for transitive closure in DAGs -- there the binary codes where of size 1 indicating whether there exists an arc between two vertices.


\subparagraph{Algebraic framework.}
In order to use coordinate-wise boolean matrix multiplication with the entries of the matrices being codewords
we first encode all $s$-$t$-earliest and all $s$-$t$-latest ${\le}k$-cuts using binary codes.
The bitwise boolean multiplication of such matrices with binary codes in its entries allows a serial combination of both $s$-$v$ cuts and $v$-$t$ cuts based on AND operations, and thus allows us to construct a solution based on the OR operation of pairwise AND operations.
We show that \emph{superimposed codes} are suitable in our case, i.e., binary codes where sets are represented as bitwise-OR of codewords of objects, and small sets are guaranteed to be encoded uniquely.
Superimposed codes provide a unique representation for sets of $k$ elements from a universe of size $\text{poly}(n)$ with codewords of length $\text{poly}(k \log n)$.
In this setting, the union of sets translates naturally to bitwise-OR of their codewords.

\subparagraph{Tensor product of codes.} To achieve our bounds, we compose several identical superimposed codes into a new binary code, so that encoding \emph{set families} with it enables us to solve the corresponding instances of \witness.
Our composition has the cost of an exponential increase in the length of the code.
Let $\mathcal{F}={F_1,\dots, F_c}$ be the set family that we wish to encode, and let $S_1,\dots,S_c$ be their superimposed codes in the form of vectors.
We construct a $c$-dimensional array $M$ where $M[i_1,\dots,i_c]=1$ iff $S_j[i_j]=1$, for each $1\leq j \leq c$.
In other words, the resulting code is the tensor product of all superimposed codes.
This construction creates enough redundancy so that enough information on the structure of the set families is preserved.
Furthermore, we can extract the encoded information from the bitwise-OR of several codewords.
The resulting code is of length $\bigO((k \log n)^{\bigO(K)})$, where $K$ is the upperbound on the allowed number of sets in each encoded set family. In our case  $K\approx 4^k$,
which results to only a logarithmic dependency on $n$ at the price of a  doubly-exponential dependency on $k$, thus making the problem tractable for small values of $k$.

\subparagraph{From slices to \witness.}
Finally, we show how the \witness can be solved using tensor product of superimposed codes.
Consider the notion of cutting the code of dimension $K$ with an axis-parallel hyperplane of dimension $K-1$.
We call this resulting shorter codeword a \emph{slice} of the original codeword.
A slice of a tensor product is a tensor product of one dimension less, or an empty set, and a slice of a bitwise-OR of tensor products is as well a bitwise-OR of tensor products (of one dimension less).
Thus, taking a slice of the bitwise-OR of the encoding of families of sets is equivalent to removing a particular set from some families and to dropping some other families completely and then encoding these remaining, reduced families.
Thus, we can design a non-deterministic algorithm, which at each step of the recursion picks $k$ slices, one slice for each element of the solution we want to output, and then recurses on the bitwise-OR of those slices, reducing the dimension by one in the process.
This is always possible, since each element that belongs to a particular solution of \witness satisfies one of the following: it either has a witnessing slice and thus it is preserved in the solution to the recursive call; or it is dense enough in the input so that it is a member of each solution and we can detect this situation from scanning the diagonal of the input codeword.
This described nondeterministic approach is then made deterministic by simply considering every possible choice of $k$ slices at each of the $K$ steps of the recursion.
This does not increase substantially the complexity of the decoding procedure, since  $\bigO(((K \cdot \text{poly}(k \log n))^k)^K)$ for $K \approx 4^k$ is still only doubly-exponential in $k$.

\section{Reducing $4$-Clique to All-Pairs Min-Cut} 
\label{sec:CLB}

In this section we prove Theorem~\ref{thm:CLB} 
by showing new reductions from the $4$-Clique problem
to $k$-bounded All-Pairs Min-Cut with unit vertex capacities. 
These reductions yield conditional lower bounds
that are much higher than previous ones, which are based on SETH,
in addition to always producing DAGs. 
Throughout this section, we will often use the term nodes for vertices. 

\begin{definition}[The $4$-Clique Problem]
Given a $4$-partite graph $G$, where $V(G)=A\cup B \cup C \cup D$ with $|A|=|B|=|C|=|D|=n$, decide whether there are four nodes $a \in A$,  $b \in B$, $c \in C$, $d \in D$ that form a clique.
\end{definition}

This problem is equivalent to the standard formulation of $4$-Clique
(without the restriction to $4$-partite graphs). 
The currently known running times are $O(n^{\omega+1})$ using matrix multiplication \cite{EisenbrandtGrandoni04}, and $O(n^4/\polylog{n})$ combinatorially \cite{huachengYu18}. The $k$-Clique Conjecture \cite{ABV15b_parsing} hypothesizes that current clique algorithms are optimal.
Usually when the $k$-Clique Conjecture is used, it is enough to assume that the current algorithms are optimal for every $k$ that is a multiple of $3$, where the known running times are $O(n^{\omega k /3})$ \cite{NP85} and $O(n^k/\polylog{n})$ combinatorially \cite{vassilevska09}, see e.g.~\cite{ABBK17, ABV15b_parsing,BW18,Chang15,LWW18}.
However, we will need the stronger assumption that one cannot improve the current algorithms for $k=4$ by any polynomial factor.
This stronger form was previously used
by Bringmann, Gr{\o}nlund, and Larsen~\cite{BringmannGL17}.

\subsection{Reduction to the Unbounded Case}
\label{sec:CLB_unbounded}

We start with a reduction to the unbounded case (equivalent to $k=n$), 
that is, we reduce to All-Pairs Min-Cut with unit node capacities
(abbreviated \APMVC, for All-Pairs Minimum Vertex-Cut).
Later (in Section~\ref{sec:CLB_unbounded})
we will enhance the construction in order to bound $k$.

\begin{lemma}\label{Lemma:technical}
Suppose \APMVC on $n$-node DAGs with unit node capacities
can be solved in time $T(n)$.
Then $4$-Clique on $n$-node graphs
can be solved in time $O(T(n) +MM(n,n))$,
where $MM(n,n)$ is the time to multiply two matrices from $\{0,1\}^{n\times n}$. 
\end{lemma}

To illustrate the usage of this lemma,
observe that an $O(n^{3.99})$-time combinatorial algorithm for \APMVC
would imply a combinatorial algorithm with similar running time for $4$-Clique.

\begin{proof}
Given a $4$-partite graph $G$ as input for the $4$-Clique problem, the graph $H$ is constructed as follows.
The node set of $H$ is the same as $G$, and we abuse notation and refer also to $V(H)$ as if it is partitioned into $A$,$B$,$C$, and $D$.
Thinking of $A$ as the set of sources and $D$ as the set of sinks, the proof will focus on the number of node-disjoint paths from nodes $a \in A$ to nodes $d \in D$.
The edges of $H$ are defined in a more special way,
see also Figure~\ref{fig:reduction} for illustration. 

\begin{figure}[h]
	\begin{minipage}{\textwidth}
		\centering\includegraphics[scale=0.4]{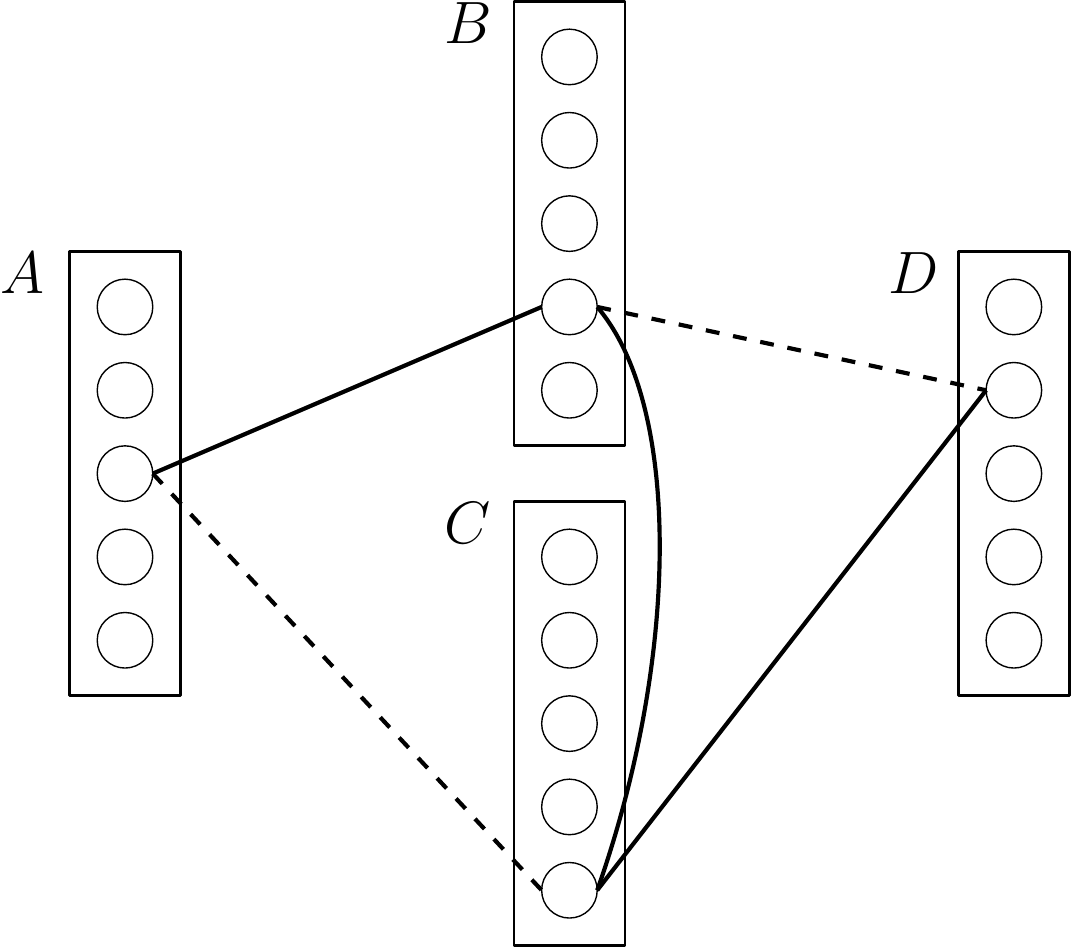}
	\end{minipage}
	\caption{An illustration of $H$ in the reduction. Solid lines between nodes represent the existence of an edge in the input graph $G$, and dashed lines represent the lack thereof. 
	}
	\label{fig:reduction}
\end{figure}

\begin{itemize}
\item (A to B) For every $a \in A, b \in B$ such that $\{a,b\} \in E(G)$, add to $E(H)$ a directed edge $(a,b)$.
\item (B to C) For every $b \in B, c \in C$ such that $\{b,c\} \in E(G)$, add to $E(H)$ a directed edge $(b,c)$.
\item (C to D) For every $c \in C, d \in D$ such that $\{c,d\} \in E(G)$, add to $E(H)$ a directed edge $(c,d)$.
\end{itemize}

The definition of the edges of $H$ will continue shortly. So far, edges in $H$ correspond to edges in $G$, and there is a (directed) path $a \to b \to c \to d$ if and only if the three (undirected) edges $\{a,b\},\{b,c\},\{c,d\}$ exist in $G$. In the rest of the construction, our goal is to make this $3$-hop path contribute to the final $a \to d$ flow \emph{if and only if} $(a,b,c,d)$ is a $4$-clique in $G$ (i.e., all six edges exist, not only those three).
Towards this end, additional edges are introduced, that make this $3$-hop path useless in case $\{a,c\}$ or $\{ b,d\}$ are not also edges in $G$. This allows ``checking'' for five of the six edges in the clique, rather than just three. The sixth edge is easy to ``check''.

\begin{itemize}
\item (A to C) For every $a \in A, c \in C$ such that $\{a,c\} \notin E(G)$, add to $E(H)$ a directed edge $(a,c)$.
\item (B to D) For every $b \in B, d \in D$ such that $\{b,d\} \notin E(G)$ in $G$, add to $E(H)$ a directed edge $(b,d)$.
\end{itemize}

This completes the construction of $H$.
Note that these additional edges imply that there is a path $a \to b \to d$ in $H$ iff $\{a,b\} \in E(G)$ and $\{b,d\} \notin E(G)$, and similarly, there is a path $a \to c \to d$ in $H$ iff $\{a,c\} \notin E(G)$ and $\{c,d\} \in E(G)$.
Let us introduce notations to capture these paths.
For nodes $a \in A, d \in D$ denote:
\begin{align*}
  B'_{a,d}
  &= \left\{ b \in B \  \mid \  \text{$\{a,b\} \in E(G)$ and $\{b,d\} \notin E(G)$} \  \right\},
  \\
  C'_{a,d}
  &= \left\{ c \in C \  \mid \  \text{$\{a,c\} \notin E(G)$ and $\{c,d\} \in E(G)$} \  \right\}.
\end{align*}

We now argue that if an \APMVC algorithm is run on $H$, enough information is received to be able to solve $4$-Clique on $G$ by spending only an additional post-processing stage of $O(n^3)$ time.

\begin{claim} \label{cl:technical}
  Let $a \in A, d \in D$ be nodes with $\{a,d\} \in E(G)$. 
  If the edge $\{a,d\}$ does not participate in a $4$-clique in $G$,
  then the node connectivity from $a$ to $d$ in $H$, denoted $NC(a,d)$,
  is exactly 
$$
NC(a,d) = | B'_{a,d} | + |C'_{a,d}|, 
$$
and otherwise $NC(a,d)$ is strictly larger.
\end{claim}


\begin{proof}[Proof of Claim~\ref{cl:technical}]
We start by observing that all paths from $a$ to $d$ in $H$ have either two or three hops.

Assume now that there is a $4$-clique $(a,b^*,c^*,d)$ in $G$,
and let us exhibit a set $P$ of node-disjoint paths from $a$ to $d$
of size $| B'_{a,d} | + |C'_{a,d}|+1$.
For all nodes $b \in B'_{a,d}$, add to $P$ the $2$-hop path $a \to b \to d$.
For all nodes $c \in C'_{a,d}$, add to $P$ the $2$-hop path $a \to c \to d$.
So far, all these paths are clearly node-disjoint.
Then, add the $3$-hop path $a \to b^* \to c^* \to d$ to $P$.
This path is node-disjoint from the rest because $b^* \notin B'_{a,d}$ (because $\{ b^*,d\} \in E(G)$) and $c^* \notin C'_{a,d}$ (because $\{a,c^*\} \in E(G)$).

Next, assume that no nodes $b \in B, c \in C$ complete a $4$-clique with $a,d$. 
Then for every set $P$ of node-disjoint paths from $a$ to $d$, there is a set $P'$ of $2$-hop node-disjoint paths from $a$ to $d$ that has the same size.
To see this, let $a \to b \to c \to d$ be some $3$-hop path in $P$. Since $(a,b,c,d)$ is not a $4$-clique in $G$ and $\{a,d\},\{a,b\},\{b,c\},\{c,d\}$ are edges in $G$, we conclude that either $\{ a,c\} \notin E(G)$ or  $\{ b,d\} \notin E(G)$.
If $\{a,c\} \notin E(G)$ then $a \to c$ is an edge in $H$ and the $3$-hop path can be replaced with the $2$-hop path $a \to c \to d$ (by skipping $b$) and one is remained with a set of node-disjoint paths of the same size.
Similarly, if $\{ b,d\} \notin E(G)$ then $b \to d$ is an edge in $H$ and the $3$-hop path can be replaced with the $2$-hop path $a \to b \to d$.
This can be done for all $3$-hop paths and result in $P'$.
Finally, note that the number of $2$-hop paths from $a$ to $d$ is exactly $| B'_{a,d} | + |C'_{a,d}|$, and this completes the proof of Claim~\ref{cl:technical}.
\end{proof}

\subparagraph{Computing the estimates.}
To complete the reduction, observe that the values $| B'_{a,d} | + |C'_{a,d}|$ can be computed for all pairs $a \in A, d \in D$ using two matrix multiplications. To compute the $| B'_{a,d} |$ values,
multiply the two matrices $M,M'$ which have entries from $\{0,1\}$,
with $M_{a,b}= 1$ iff $\{a,b\} \in E(G) \cap A \times B$
and $M'_{b,d}= 1$ iff $\{b,d\} \notin E(G) \cap B \times D$.
Observe that $| B'_{a,d} |$ is exactly $(M \cdot M')_{a,d}$.
To compute $|C'_{a,d}|$, multiply $M,M'$ over $\{0,1\}$ where $M_{a,c}= 1$ iff $\{a,c\} \notin E(G) \cap A \times C$ and $M'_{c,d}= 1$ iff $\{c,d\} \in E(G) \cap C \times D$.

After having these estimates and computing \APMVC on $H$, it can be decided whether $G$ contains a $4$-clique in $O(n^2)$ time as follows. Go through all edges $\{a,d\} \in E(G) \cap A \times D$ and decide whether the edge participates in a $4$-clique by comparing $| B'_{a,d} | + |C'_{a,d}|$ to the node connectivity $NC(a,d)$ in $H$. By the above claim, an edge $\{a,d\}$ with $NC(a,d) > | B'_{a,d} | + |C'_{a,d}|$ is found if and only if there is a $4$-clique in $G$.
The total running time is $O(T(n) +MM(n))$,
which completes the proof of Lemma~\ref{Lemma:technical}.
\end{proof}

\subsection{Reduction to the $k$-Bounded Case}
\label{sec:CLB_bounded}

Next, we exploit a certain versatility of the reduction
and adapt it to ask only about min-cut values (aka node connectivities)
that are smaller than $k$.
In other words, we will reduce to
the $k$-bounded version of All-Pairs Min-Cut with unit node capacities 
(abbreviated \kAPMVC, for $k$-bounded All-Pairs Minimum Vertex-Cut). 
Our lower bound improves on the $\Omega(n^\omega)$ conjectured lower bound
for Transitive Closure as long as $k=\omega(n^{1/2})$.

\begin{lemma}\label{Lemma:technicalLow}
Suppose \kAPMVC on $n$-node DAGs with unit node capacities
can be solved in time $T(n,k)$.
Then $4$-Clique on $n$-node graphs
can be solved in time $O(\frac{n^2}{k^2} \cdot T(n,k) +MM(n) )$,
where $MM(n,n)$ is the time to multiply two matrices from $\{0,1\}^{n\times n}$. 
\end{lemma}
\begin{proof}[Proof of Lemma~\ref{Lemma:technicalLow}]
Given a $4$-partite graph $G$ as in the definition of the $4$-Clique problem, $O(n^2/k^2)$ graphs $H$ are constructed in a way that is similar to the previous reduction, and an algorithm for \kAPMVC is called on each of these graphs.
Assume w.l.o.g. that $k$ divides $n$ and partition the sets $A,D$ arbitrarily to sets $A_1,\ldots,A_{n/k}$ and $D_1,\ldots, D_{n/k}$ of size $k$ each. For each pair of integers $i,j \in [n/k]$, generate one graph $H_{ij}$ by restricting the attention to the nodes of $G$ in $A_i, B,C,D_j$ and looking for a $4$-clique only there.

Let us fix a pair $i,j\in [n/k]$ and describe the construction of $H_{ij}$.
To simplify the description, let us omit the subscripts $i,j$, referring to this graph as $H$, and think of $G$ as having four parts $A,B,C,D$,
where $A$ and $D$ are in fact $A_i,D_j$ and are therefore smaller:
$|A|=|D|=k$ and $|B|=|C|=n$.

The nodes in $H$ are partitioned into four sets $A',B,C,D'$, where the sets $B,C$ are the same as in $G$.
For the nodes in $A,D$ in $G$, multiple copies are created in $H$.
For all integers $x \in [n/k]$ and node $a \in A$ in $G$, add a node $a_x$ to $A'$ in $H$.
Similarly, for all $x \in [n/k]$ and node $d \in D$, add a node $a_x$ to $A'$.
Note that $H$ contains $O(n)$ nodes.

To define the edges, partition the nodes in $B$ and $C$ arbitrarily to sets $B_1,\ldots,B_{n/k}$ and $C_1,\ldots,C_{n/k}$ of size $k$.
Now, the edges are defined in a similar way to the previous proof, except each $a_x$ is connected only to nodes in $B_{x}$, and each $d_y$ is connected only to nodes in $C_y$. More formally:

\begin{itemize}
\item (A to B) For every $a_x \in A', b \in B_x$ such that $\{a,b\} \in E(G)$, add to $E(H)$ a directed edge $(a_x,b)$.
\item (B to C) For every $b \in B, c \in C$ such that $\{b,c\} \in E(G)$, add to $E(H)$ a directed edge $(b,c)$.
\item (C to D) For every $c \in C_y, d_y \in D'$ such that $\{c,d\} \in E(G)$, add to $E(H)$ a directed edge $(c,d_y)$.
\item (A to C) For every $a_x \in A', c \in C$ such that $\{a,c\} \notin E(G)$, add to $E(H)$ a directed edge $(a_x,c)$.
\item (B to D) For every $b \in B, d_y \in D'$ such that $\{b,d\} \notin E(G)$, add to $E(H)$ a directed edge $(b,d_y)$.
\end{itemize}

This completes the construction of $H$.
The arguments for correctness follow the same lines as in the previous proof.
For nodes $a_x \in A', d_y \in D'$ denote:
\begin{align*}
  B'_{a_x,d_y}
  &= \left\{ b \in B_x \  \mid \  \text{$\{a,b\} \in E(G)$ and $\{b,d\} \notin E(G)$} \  \right\},
  \\
  C'_{a_x,d_y}
  &= \left\{ c \in C_y \  \mid \  \text{$\{a,c\} \notin E(G)$ and $\{c,d\} \in E(G)$} \  \right\}.
\end{align*}

\begin{claim} \label{cl:technicalLow}
  Let $a_x \in A', d_y \in D'$ be nodes with $\{a,d\} \in E(G)$.
  If the edge $\{a,d\}$ does not participate in a $4$-clique in $G$
  together with any nodes in $B_x \cup C_y$,
  then the node connectivity from $a_x$ to $d_y$ in $H$, denoted $NC(a_x,d_y)$,
  is exactly
$$
NC(a_x,d_y) = | B'_{a_x,d_y} | + |C'_{a_x,d_y}|
$$
and otherwise $NC(a_x,d_y)$ is strictly larger. 
\end{claim}

\begin{proof}[Proof of Claim~\ref{cl:technicalLow}]
The proof is very similar to the one in the previous reduction.

We start by observing that all paths from $a_x$ to $d_y$ in $H$ can have either two or three hops.

For the first direction, assuming that there is a $4$-clique $(a,b^*,c^*,d)$ in $G$ with $b^* \in B_x, c^* \in C_y$, we show a set $P$ of node-disjoint paths from $a_x$ to $d_y$ of size $| B'_{a_x,d_y} | + |C'_{a_x,d_y}|+1$.
For all nodes $b \in B'_{a_x,d_y}$, add the $2$-hop path $a_x \to b \to d_y$ to $P$.
For all nodes $c \in C'_{a_x,d_y}$, add the $2$-hop path $a_x \to c \to d_y$ to $P$.
So far, all these paths are clearly node-disjoint.
Then, add the $3$-hop path $a_x \to b^* \to c^* \to d_y$ to $P$.
This path is node-disjoint from the rest because $b^* \notin B'_{a_x,d_y}$ (because $\{ b^*,d\} \in E(G)$) and $c^* \notin C'_{a_x,d_y}$ (because $\{a,c^*\} \in E(G)$).

For the second direction, assume that there do not exist nodes $b \in B_x, c \in C_y$ that complete a $4$-clique with $a,d$.
In this case, for every set $P$ of node-disjoint paths from $a_x$ to $d_y$, there is a set $P'$ of $2$-hop node-disjoint paths from $a_x$ to $d_y$ that has the same size.
To see this, let $a_x \to b \to c \to d_y$ be some $3$-hop path in $P$. Since $(a,b,c,d)$ is not a $4$-clique in $G$ and $\{a,d\},\{a,b\},\{b,c\},\{c,d\}$ are edges in $G$, it follows that either $\{ a,c\} \notin E(G)$ or  $\{ b,d\} \notin E(G)$.
If $\{a,c\} \notin E(G)$ then $a_x \to c$ is an edge in $H$ and the $3$-hop path can be replaced with the $2$-hop path $a_x \to c \to d_y$ (by skipping $b$) and one is remained with a set of node-disjoint paths of the same size.
Similarly, if $\{ b,d\} \notin E(G)$ then $b \to d_y$ is an edge in $H$ and the $3$-hop path can be replaced with the $2$-hop path $a_x \to b \to d_y$.
This can be done for all $3$-hop paths and result in $P'$.
Finally, note that the number of $2$-hop paths from $a_x$ to $d_y$ is exactly $| B'_{a_x,d_y} | + |C'_{a_x,d_y}|$,
and this completes the proof of Claim~\ref{cl:technicalLow}.
\end{proof}

This claim implies that in order to determine whether a pair $a \in A, d \in D$ participate in a $4$-clique in $G$ is it is enough to check whether $ \sum_{x,y \in [n/k]} NC(a_x,d_y)$ is equal to $\sum_{x,y \in [n/k]} | B'_{a_x,d_y} | + |C'_{a_x,d_y}|$. Note that the latter is equal to $|B'_{a,d}|+|C'_{a,d}|$ according to the notation in the previous reduction:
\begin{align*}
  B'_{a,d}
  &= \left\{ b \in B \  \mid \  \text{$\{a,b\} \in E(G)$ and $\{b,d\} \notin E(G)$} \  \right\},
  \\
  C'_{a,d}
  &= \left\{ c \in C \  \mid \  \text{$\{a,c\} \notin E(G)$ and $\{c,d\} \in E(G)$} \  \right\}.
\end{align*}

\subparagraph{Computing the estimates}
To complete the reduction, observe that the values $| B'_{a,d} | + |C'_{a,d}|$ can be computed for all pairs $a \in A, d \in D$ (for all sub-instances $i,j$) using two matrix products, just like in the previous reduction.

After having these estimates and computing \APMVC on $H$, it can be decided whether $G$ contains a $4$-clique in $O(k^2\cdot n/k)$ time, for each sub-instance $i,j$, as follows. Go through all edges $\{a,d\} \in E(G) \cap A_i \times D_j$ and check whether the edge participates in a $4$-clique by comparing this value $| B'_{a,d} | + |C'_{a,d}|$ to the node connectivities $\sum_{x,y \in [n/k]} NC(a_x,d_y)$ in $H$. By the above claim, one can find an edge $\{a,d\}$ with $\sum_{x,y \in [n/k]} NC(a_x,d_y) > | B'_{a,d} | + |C'_{a,d}|$ if and only if there is a $4$-clique in $G$.
The total running time is $O(\frac{n^2}{k^2} \cdot T(n,k) + MM(n))$,
which completes the proof of Lemma~\ref{Lemma:technicalLow}. 
\end{proof}

\begin{proof}[Proof of Theorem~\ref{thm:CLB}] 
Assume there is an algorithm that solves \kAPMVC
in time $O((n^{\omega-1}k^2)^{1-\varepsilon})$.
Then by Lemma~\ref{Lemma:technicalLow} there is an algorithm
that solves $4$-Clique
in time $=O(\frac{n^2}{k^2} \cdot (n^{\omega-1}k^2)^{1-\varepsilon}+MM(n)) \leq
O(n^{\omega+1-\varepsilon'})$, for some $\varepsilon'>0$.
The bound for combinatorial algorithms is achieved similarly.
\end{proof}

\section{Randomized Algorithms for General Digraphs}
\label{app:general}

In this section we develop faster randomized algorithms for the following problems.
Given a digraph $G=(V,E)$ with unit vertex capacities,
two subsets $S, T \subseteq V$ and parameter $k>0$,
find all $s \in S, t \in T$ for which the \stmincut value is less than $k$
and report their min-cut value. 
This problem is called \kSTMVC, and if $S=T=V$, it is called \kAPMVC.
This is done by showing that the framework of Cheung et al.~\cite{CheungLL13} can be applied faster to unit vertex-capacitated graphs.
Before providing our new algorithmic results
(in Theorem~\ref{Thm:Vertices} and Corollary~\ref{Cor:Vertices}),
we first give some background on network coding
(see~\cite{CheungLL13} for a more comprehensive treatment).

\paragraph*{The Network-Coding Approach.} 
Network coding is a novel method for transmitting information in a network. As shown in a fundamental result~\cite{AhlswedeCLY00}, if the edge connectivity from the source $s$ to each sink $t_i$ is $\geq k$, then $k$ units of information can be shipped to all sinks simultaneously by performing encoding and decoding at the vertices. This can be seen as a max-information-flow min-cut theorem for multicasting, for which an elegant algebraic
framework has been developed for constructing efficient network coding schemes~\cite{li2003linear, KoetterM03}.
These techniques were used in~\cite{CheungLL13} to compute edge connectivities, and below we briefly recap their method and notation.

Given a vertex $s$ from which we need to compute the maximum flow to all other vertices in $G$, define the following matrices over a field $\mathbb{F}$.
\begin{itemize}
\item $F_{d\times m}$ is a matrix whose $m$ columns are $d$-dimensional global encoding vectors of the edges, with $d=\deg_G^{out}(s)$.
\item $K_{m\times m}$ is a matrix whose entry $(e_1,e_2)$ corresponds to the local encoding coefficient $k_{e_1,e_2}$ which is set to a random value from the field $\card{\mathbb{F}}=O(m^c)$ if $e_1$'s head is $e_2$'s tail, and to zero otherwise.
\item $H_{d\times m}$ is a matrix whose columns are $(\overrightarrow{e_1},\dots, \overrightarrow{e_d},\overrightarrow{0},\dots,\overrightarrow{0})$ ($\overrightarrow{e_i}$ is in the column corresponding to $e_i$) where the column vector $\overrightarrow{e_i}$ is the $i$th standard basis vector and $e_1,\dots,e_d$ are the edges outgoing of $s$.
\end{itemize}

The global encoding vectors $F$ are defined such that $F=FK+H$, and then by simple manipulations the equation $F=H(I-K)^{-1}$ is achieved (so multiplying by $H$ simply picks rows of $(I-K)^{-1}$ that correspond to the edges outgoing of $s$). The algorithm utilizes this by first computing $(I-K)^{-1}$ in time $O(m^{\omega})$, and then for every source $s$ and sink $t$ computing the rank of the submatrix corresponding to rows $\delta^{out}(s)$ and columns $\delta^{in}(t)$ in time $O(m^2n^{\omega-2})$, and the overall time is $O(m^{\omega})$ since $m\geq n$. Notice that even if we only care about the maximum flow between given sets of sources and targets $S,T$,
the bottleneck is that the matrix $(I-K)$ has $m^2$ entries, potentially most are non-zeroes, which must be read to compute $(I-K)^{-1}$.

\paragraph*{Our Algorithmic Results.} 
\begin{lemma}\label{Lemma:Technical}
  \kSTMVC can be solved in randomized time
  $O\Big( n^{\omega}
  + \sum_{s\in S} \sum_{t\in T} \deg_G^{out}(s)
\\ \deg_G^{in}(t)^{\omega-1} \Big)$,
  where $\deg_G^{out}(u)$ and $\deg_G^{in}(u)$ denote the out-degree and the in-degree, respectively, of vertex $u$ in the input graph $G$.
\end{lemma}
\begin{proof}
We consider global encoding vectors in the vertices rather than in the edges in the natural way, namely, the coefficients are non-zero for every pair of adjacent vertices (rather than adjacent edges), and for a source $s$ and a sink $t$ we compute the rank of the submatrix of $(I-K)^{-1}$ whose rows correspond to the vertices $N^{out}(s)$ and columns correspond to $N^{in}(t)$.
The running time is dominated by inverting the matrix $(I-K)_{n\times n}$ and computing the rank of the relevant submatrices, that is $O(n^{\omega}+\sum_{s\in S}\sum_{t\in T}\deg_G(s)\deg_G(t)^{\omega-1})$, as required. Notice that by considering vertices rather than edges, the bottleneck moves from computing $(I-K)^{-1}$ to computing the rank of the relevant submatrices.

To prove the correctness, we argue that Theorem $2.1$ from~\cite{CheungLL13} holds also here (adjusted to node-capacities). Part $1$ in their proof clearly holds also here, so we focus on the second part, which in~\cite{CheungLL13} shows that the edge connectivity from $s$ to $t$, denoted $\lambda_{s,t}$, is equal to the rank of the matrix $M_{s,t}$ of size $\deg_G(s)\times \deg_G(t)$ comprising of the global encoding vectors on the edges incoming to $t$ as its columns. Here, we denote the vertex connectivity from $s$ to $t$ by $\kappa_{s,t}$, and the corresponding matrix $M_{s,t}^{vertices}$, and we show that their proof can be adjusted to show $\rank(M_{s,t}^{vertices})=\kappa_{s,t}$, as required. First, $\rank(M_{s,t}^{vertices})\leq \kappa_{s,t}$ as instead of considering an edge-cut $(S,T)$ and claiming that the global encoding vector on each incoming edge of $t$ is a linear combination of the global encoding vectors of the edges in $(S,T)$, we consider a node-cut $(S^{vertices},C^{vertices},T^{vertices})$, and similarly claim that the global encoding vectors on each vertex with an edge to $t$ is a linear combination of the global encoding vectors in $C^{vertices}$, and the rest of the proof follows.
For the second part, we argue that $\rank(M_{s,t}^{vertices}) \geq \kappa_{s,t}$. The main proof idea from~\cite{CheungLL13} that the rank does not increase if we restrict our attention to a subgraph holds here too, only that we use vertex disjoint paths as the subgraph to establish the rank. 
\end{proof}

\begin{theorem} \label{Thm:Vertices}
  \kSTMVC can be solved in randomized time
  $O\Big( \big(n+((\card{S}+\card{T})k)\big)^{\omega}+\card{S}\card{T}k^{\omega} \Big)$.
\end{theorem}
\begin{proof}
In order to use Lemma~\ref{Lemma:Technical} to prove Theorem~\ref{Thm:Vertices}, we need to decrease the degree of sources $S$ and sinks $T$. Thus, for every source $s$ we add a layer of $k$ vertices $L_s$ and connect $s$ to all the vertices in $L_s$ which in turn are connected by a complete directed bipartite graph to the set of vertices $N^{out}(s)$, directed away from $L_s$. Similarly, for every sink $t\in T$ we add a layer of $k$ vertices $L_t$ and connect to $t$ all the vertices in $L_t$, which in turn are connected by a complete directed bipartite graph from the set of vertices $N^{in}(t)$, directed away from $N^{in}(t)$.
Note that all flows of size $\leq k-1$ are preserved, and flows of size $\geq k$ become $k$. 
This incurs an additive term $(\card{S}+\card{T})k$ in the dimension of the matrix inverted, 
and altogether we achieve a running time of $\bigO((n+(\card{S}+\card{T})k)^{\omega}+\card{S}\card{T}k^{\omega})$, as required.
\end{proof}

As an immediate corollary we have the following.
\begin{corollary}\label{Cor:Vertices}
\kAPMVC can be solved in randomized time $\bigO((n k)^{\omega})$.
\end{corollary}

\section{Structure of Cuts}
\label{sec:cuts-structure}

In this section, we study the dependence of the latest $s$-$t$ cuts on the $s$-$v$ cuts and the $v$-$t$ cuts, for all vertices  $v\notin\{s,t\}$.
None of the results contained in this section rely on the input graph $G$ being acyclic.

First, we present some basic relations between flow and extremal $s$-$t$ min-cuts.
To avoid repetitions below, we state some results only for \emph{latest} cuts. However, all of them naturally extend to earliest cuts.

%
\begin{lemma}[Latest $s$-$t$ min-cut {\cite[Theorem 5.5]{ford1962flows}}]
	\label{lem:mincutinclusion}
	For any directed graph $G = (V,A)$, any maximum $s$-$t$ flow $f$ defines the same set of $t$-reaching vertices
	$T_{s,t}$ and thus defines an $s$-$t$ cut $M=A \cap (\lbar{T_{s,t}} \times T_{s,t})$, with $T_{s,t} = \{x \in V \mid \exists\text{ $x$-$t$ path in residual graph of $G$ under flow $f$}\}$.
	For any $s$-$t$ min-cut $M'$, we have $M' \le M$.
\end{lemma}
Maximum flows are not necessarily unique, but Lemma~\ref{lem:mincutinclusion} shows that the $t$-reaching cut $M$ is.

\begin{corollary}
	\label{lem:latestmincut}
	For any digraph $G$ and vertices $s$ and $t$, the latest $s$-$t$ min-cut is unique.
\end{corollary}

Next,we introduce some notation for sets of extremal cuts and their transitive order.
Then, building on the uniqueness of the latest min-cut (Corollary~\ref{lem:latestmincut}), we constructively define an operation which we call \emph{arc replacement} in an $s$-$t$-latest cut.
We refer to Figures~\ref{fig:arc_replacement_examples} and \ref{fig:arc_replacement_definition} for illustrations.

\subparagraph{Transitive reduction.}
By $\mathcal{F}_{s,t}$ we denote the set of $s$-$t$-latest cuts, by $\mathcal{F}_{s,t}^k$ the set of $s$-$t$-latest $k$-cuts. Sets of earliest cuts are denoted by $\mathcal{E}_{s,t}$ and $\mathcal{E}_{s,t}^k$ respectively. We also denote $\mathcal{F}_{s,t}^{\le k} = \bigcup_{i=1}^k \mathcal{F}_{s,t}^i$.
Since `$>$', the partial order on cuts, is a transitive relation, we can consider its transitive reduction. We say that $M' \in \mathcal{F}_{s,t}$ is \emph{immediately later} than $M \in \mathcal{F}_{s,t}$, if $M' > M$ and there is no $M'' \in \mathcal{F}_{s,t}$ such that $M' > M'' > M$.

\begin{definition}[Arc replacement]
	Given $M \in \mathcal{F}_{s,t}$ and $a = (u,v) \in M$, let $G' = (V', A')$ be a copy of $G$, where all vertices in $S_M \cup \{v\}$ are contracted to vertex $s$.
	We call the unique latest $s$-$t$ min-cut $M^*$ in $G'$ the \emph{arc replacement of $M$ and $a$ in $G$} (or say that it does not exist if $s$ and $t$ got contracted into the same vertex in $G'$ whenever $v=t$).
\end{definition}

Note that the arcs $A'$ in $G'$ correspond to a subset of the arcs $A$ in $G$ as we think of the contraction as a relabeling of some of the endpoints without changing any identifiers.
We note that a similar operation with respect to latest min-cuts was used by Baswana, Choudhary, and Roditty~\cite{BaswanaCR16}, but in a different way. Given a digraph $G$ and two vertices $s,t$, Baswana et al.~\cite{BaswanaCR16} use an operation to compute a set $A_t$ of incoming arcs to $t$ with the following property: Let $G'$ be the subgraph of $G$ where the only arcs entering $t$ are the arcs in $A_t$. 
Then, there exist $k$ arc-disjoint paths from $s$ to $t$ in $G$ iff there exist $k$ arc-disjoint paths from $s$ to $t$ in $G'$.
Here, we use the arc replacement operation to relate all $s$-$t$-latest cuts, as we next show.

\begin{figure}[tb]
	\begin{minipage}{.5\textwidth}
		\centering\includegraphics[scale=0.6]{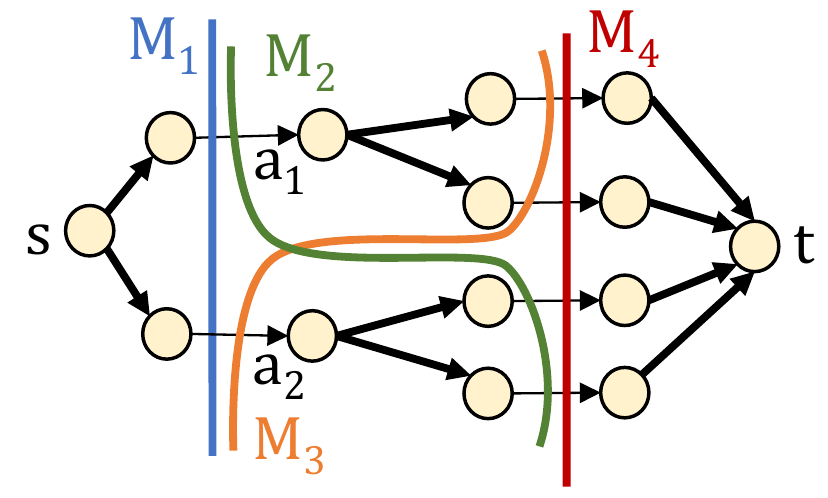}
	\end{minipage}
	\begin{minipage}{.5\textwidth}
		\centering\includegraphics[scale=0.6]{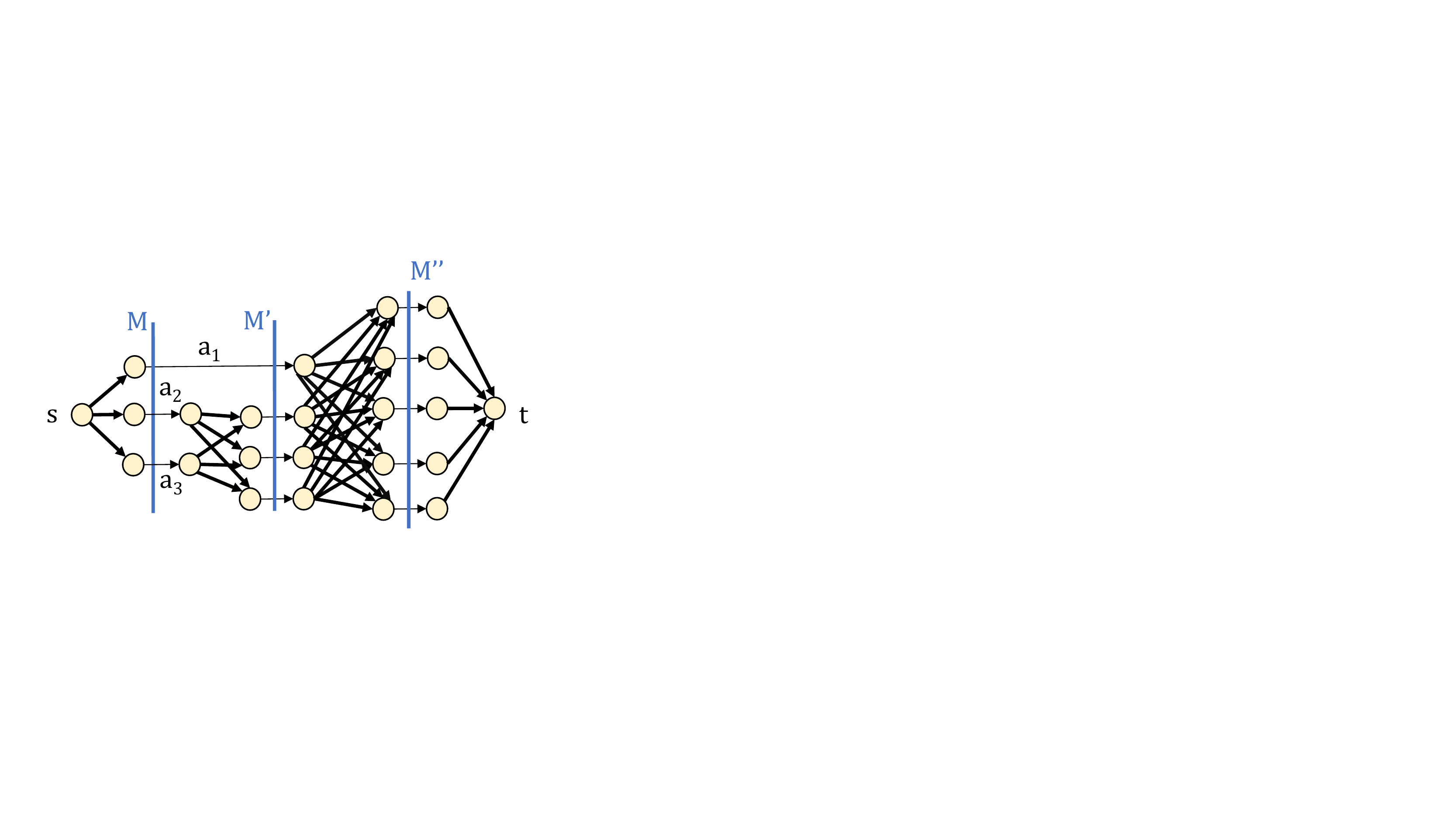}
	\end{minipage}
	\caption{Examples of digraphs with several $s$-$t$ cuts. Bold arcs represent parallel arcs which would be too expensive to cut.
		(left) We have $M_2 > M_1$, $M_3 > M_1$, $M_4 > M_2$ and $M_4 > M_3$ and all four cuts are $s$-$t$-latest. $M_2$ and $M_3$ are incomparable, neither of them is later than the other.
		$M_2$ is the arc replacement of $M_1$ and $a_2$, $M_3$ is the arc replacement of $M_1$ and $a_1$. $M_4$ is both the arc replacement of $M_2$ and $a_1$ and of $M_3$ and $a_2$.
		(right) The cut $M = \{a_1, a_2, a_3\}$ is the only $s$-$t$ min-cut.
		$M'$ is the only $s$-$t$-latest cut of size $4$ and $M''$ is the only $s$-$t$-latest cut of size $5$.
		Note that we have $M'' > M' > M$.
		$M''$ is the arc replacement of $M$ and $a_1$, while
		$M'$ is the arc replacement of $M$ and $a_2$ and also of $M$ and $a_3$.
	}
	\label{fig:arc_replacement_examples}
\end{figure}

\begin{figure}
	\begin{minipage}{\textwidth}
		\centering\includegraphics[scale=0.6]{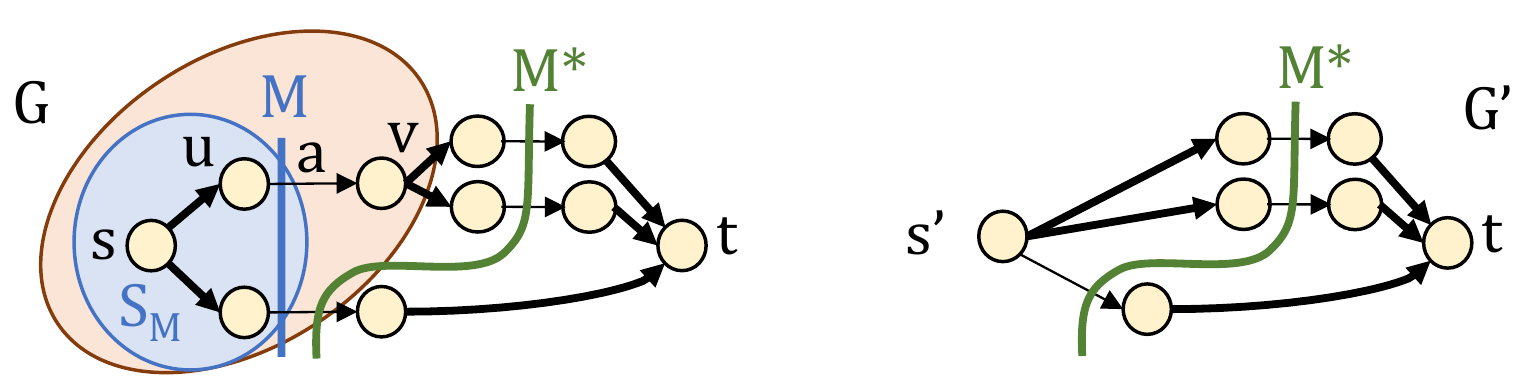}
	\end{minipage}
	\caption{(left) The original graph $G$, with bold arcs representing parallel arcs and the $s$-$t$-latest (min-)cut $M$.
		To replace $a = (u,v) \in M$ in $G$, the source side $S_M$ (blue set) together $v$ (red set) gets contracted into a new source vertex $s'$.
		(right) The graph $G'$ after the contraction.
		The latest $s'$-$t$ min-cut $M^*$ in $G'$ corresponds to the arc replacement of $M$ and $a$ in $G$.
	\label{fig:arc_replacement_definition}
	}
\end{figure}

\begin{lemma}
For any $s$-$t$-latest cuts $M$ and $M'$, if $M'$ is immediately later than $M$ then $M'$ is the arc replacement of $M$ and arc $a$, where $a$ can be any arc in $M \setminus M'$.
\label{lem:arcreplacement}
\end{lemma}
\begin{figure}
\begin{minipage}{.5\textwidth}
\centering\includegraphics[scale=0.6]{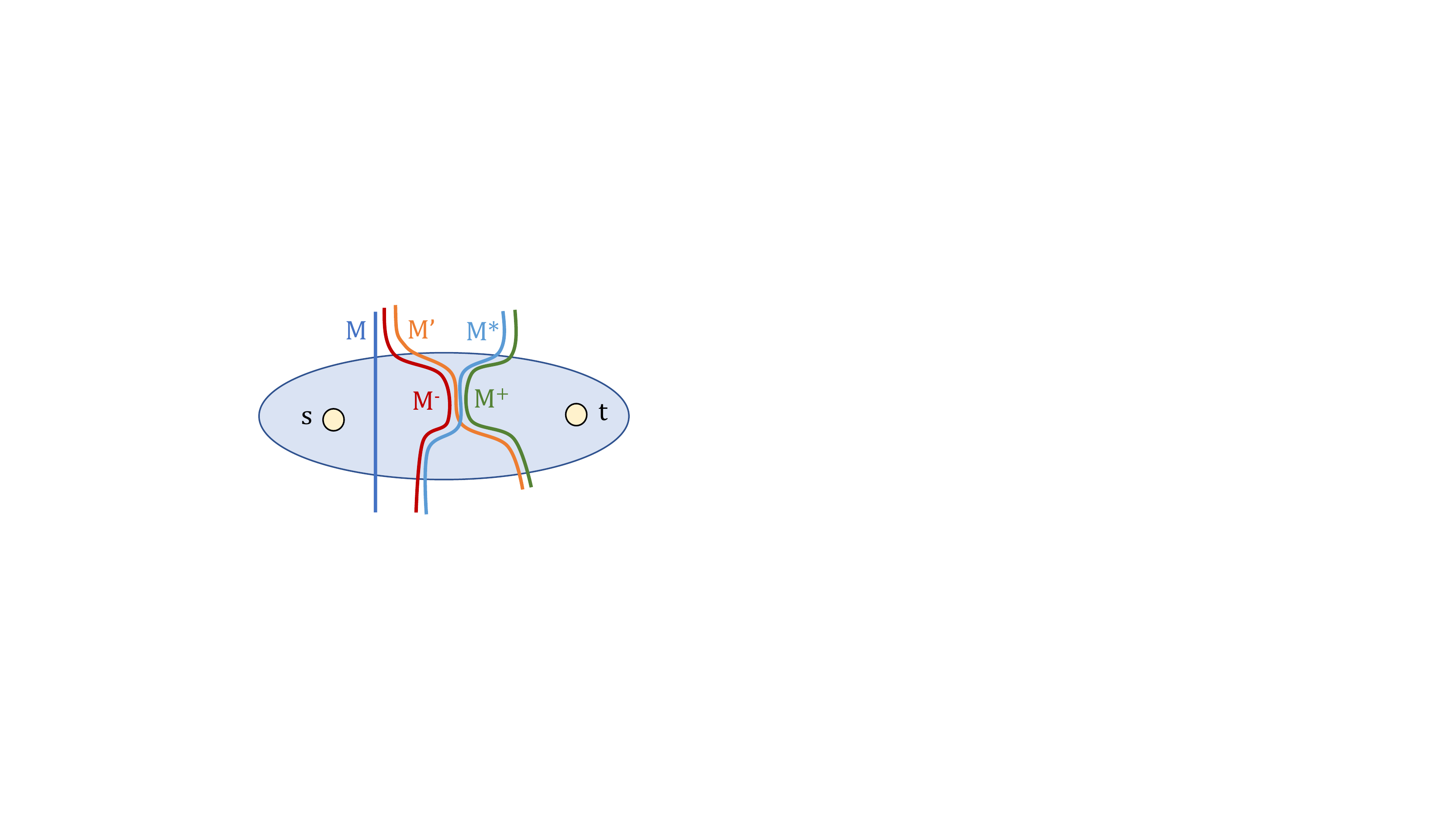}
\end{minipage}
\begin{minipage}{.5\textwidth}
\centering\includegraphics[scale=0.6]{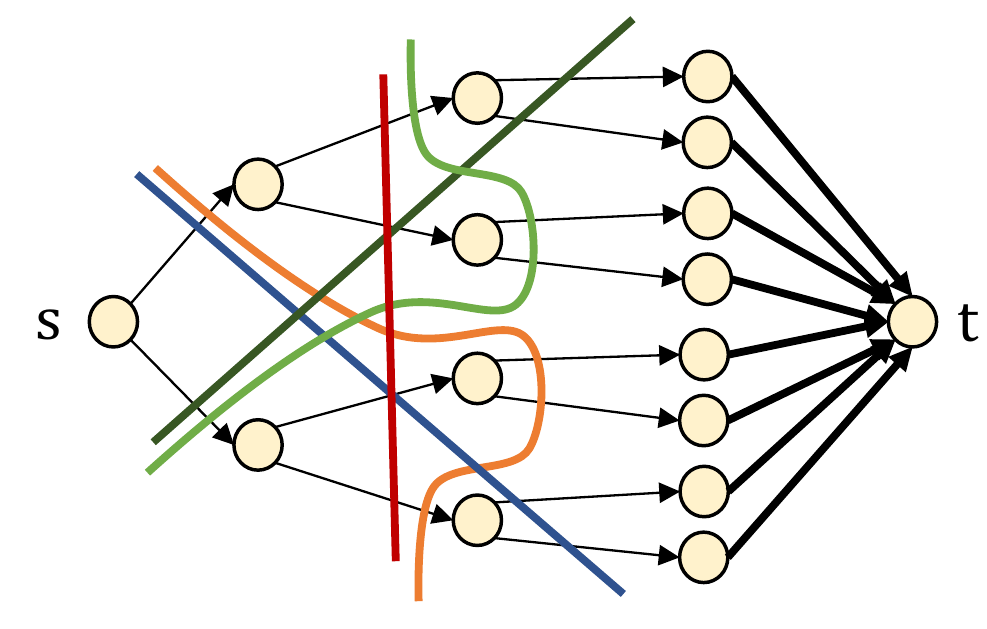}
\end{minipage}
\caption{(left) The cuts involved in the proof of Lemma~\ref{lem:arcreplacement}. If $M'$ and $M^*$ differ, $M^+$ would be both later and not larger than $M'$, hence $M'$ could not be $s$-$t$-latest.
(right) A tight example for Lemma~\ref{lem:catalan} (bold arcs represent parallel arcs): A digraph with all $C_3 = 5$ many $s$-$t$-latest $4$-cuts.}
\label{fig:arc_replacement_proofs}
\end{figure}
\begin{proof}
For a given $s$-$t$-latest cut $M$, take any arc $a^* \in M \setminus M'$ and build the corresponding arc replacement graph $G'$.
Let $M^*$ be the latest $s$-$t$ min-cut in $G'$ (which is unique by Corollary~\ref{lem:latestmincut}).
As we have $a' \geq M$ for all $a' \in A'$, we get $M' > M$ in $G$.
Also, $M^* \in \mathcal{F}_{s,t}$ (all cuts later than $M^*$ have the same cardinality in $G$ and in $G'$ and hence cannot be smaller).
The cut $M'$ has size $|M'|$ in $G'$ as all arcs $a' \in M'$ satisfy $a' \geq M$ (as $M' \geq M$) and $a' \neq a^*$ and thus are not contracted in $G'$.
By minimality of $M^*$, we get $|M^*| \leq |M'|$.

For the sake of reaching a contradiction, let us assume that $M^* \neq M'$.
$M^*$ and $M'$ can not be comparable: $M' > M^*$ would contradict $M'$ being immediately later than $M$, and $M^* > M'$ would contradict $M'$ being in $\mathcal{F}_{s,t}$.
We now define two auxiliary cuts (see Figure~\ref{fig:arc_replacement_proofs} (left) for an illustration):
\begin{align*}
M^+ &= \{a \in M' \cup M^* \mid a \geq M' \text{ and } a \geq M^*\}\\
M^- &= \{a \in M' \cup M^* \mid M' \geq a \text{ and } M^* \geq a\}
\end{align*}
As $M^+$ corresponds to $(X, \lbar{X})$ with $\lbar{X} = \lbar{T_{M'}} \cap \lbar{T_{M^*}}$
and $M^-$ to $(Y, \lbar{Y})$ with $\lbar{Y} = \lbar{T_{M'}} \cup \lbar{T_{M^*}}$,
these two arc sets really correspond to minimal $s$-$t$-cuts.
We have $M^+ \geq M', M^* \geq M^-$ and $|M^-| + |M^+| = |M'| + |M^*|$, which combined with $|M^*| \leq |M^+|, |M^-|, |M'|$ gives $|M^+| \leq |M'|$.
However, this contradicts $M' \in \mathcal{F}_{s,t}$ as $M^+$ would both be later and not larger than $M'$.
\end{proof}
Note that the reverse direction does not hold: Some arc replacements result in cuts that are not immediately later, as illustrated by $M''$ in Figure~\ref{fig:arc_replacement_examples} (right). The following lemma extends the uniqueness of min-cuts (Lemma~\ref{lem:latestmincut}) to bounding the number of $s$-$t$-latest cuts in general.
\begin{lemma}[Theorem 8.11 in \cite{cygan2015parameterized}]
\label{lem:catalan}
For any $k \ge 1$ there are at most
$C_{k-1} = \frac{1}{k}\binom{2k-2}{k-1}$ $s$-$t$-latest $k$-cuts, and at most $4^k$ $s$-$t$-latest $\le$$k$-cuts.
\end{lemma}
Note that the bound of $C_{k-1}$ is tight: consider a full binary tree with arcs directed away from the root $s$, and an extra vertex $t$ with incoming arcs from every leaf of the tree. If the tree is large enough, any $s$-$t$ $k$-cut is latest and corresponds to binary subtree with $k$ leaves. We refer to Figure~\ref{fig:arc_replacement_proofs} (right) for an example.

\begin{definition}[Arc split \cite{icalp2017GGIPU}]
\label{def:arc-split}
For $G=(V,A)$, let $A_1, A_2$ be a partition of its arc set $A$, $A=A_1\cup A_2$.
We say that a partition is an \emph{arc split} if there is no triplet of vertices $x,y,z$ in $G$ such that $(x,y) \in A_2$ and $(y,z) \in A_1$.
\end{definition}
Informally speaking, under such a split, any path in $G$ from a vertex $u$ to a vertex $v$
consists of a sequence of arcs from $A_1$ followed by a sequence of arcs from $A_2$ (as a special case, any of those sequences can be empty). In acyclic graph arc split is easily obtained from a topological order, e.g. by partitioning $V$ into prefix of order $V_1$ and suffix $V_2$, and assigning $A \cap (V_1 \times V_1)$ to $A_1$, $A \cap (V_2 \times V_2)$ to $A_2$ and arcs from $A \cap (V_1 \times V_2)$ arbitrarily.

\begin{figure}
\begin{minipage}{.4\textwidth}
\centering\includegraphics[scale=0.6]{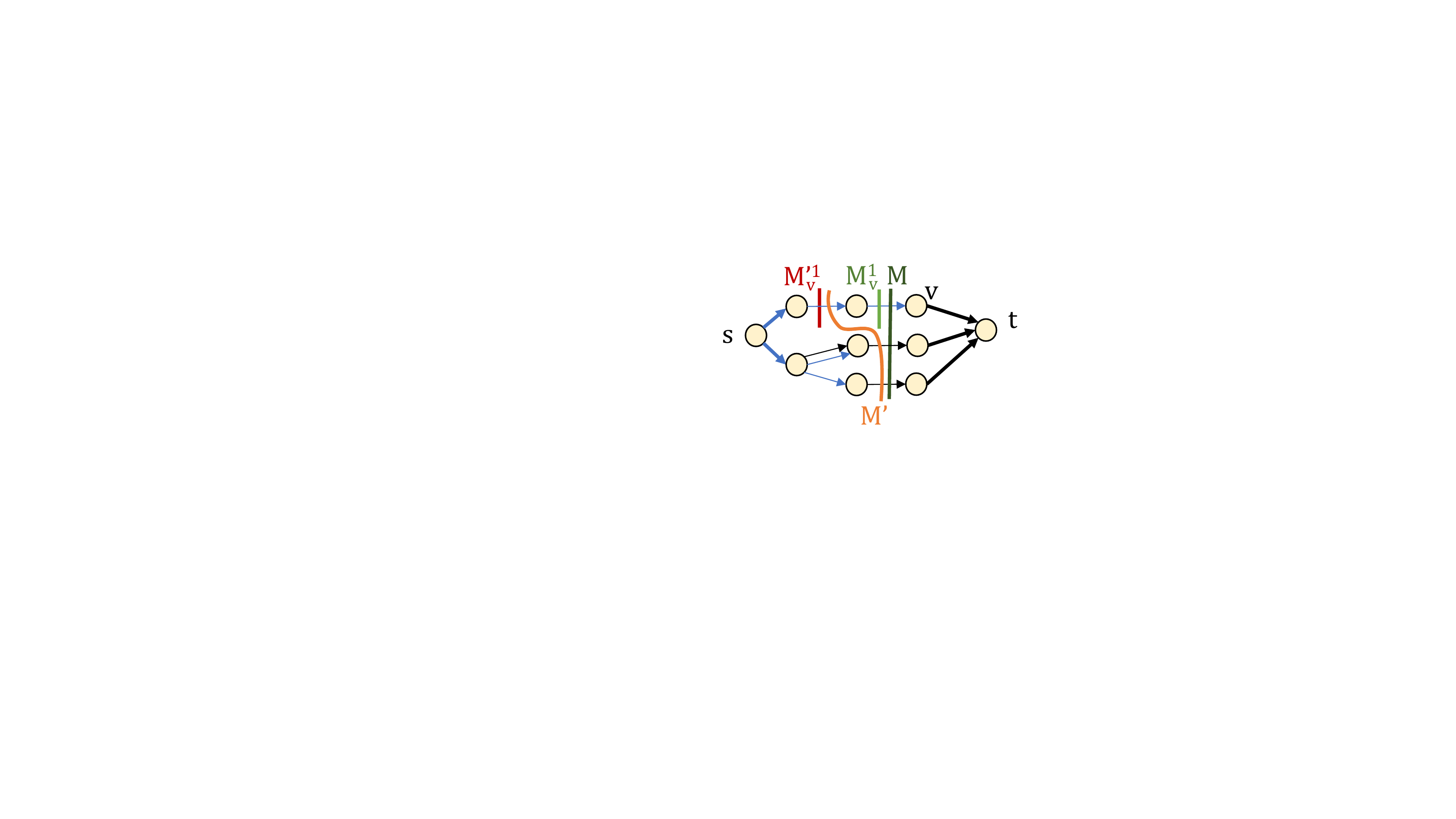}
\end{minipage}
\begin{minipage}{.55\textwidth}
\centering\includegraphics[scale=0.6]{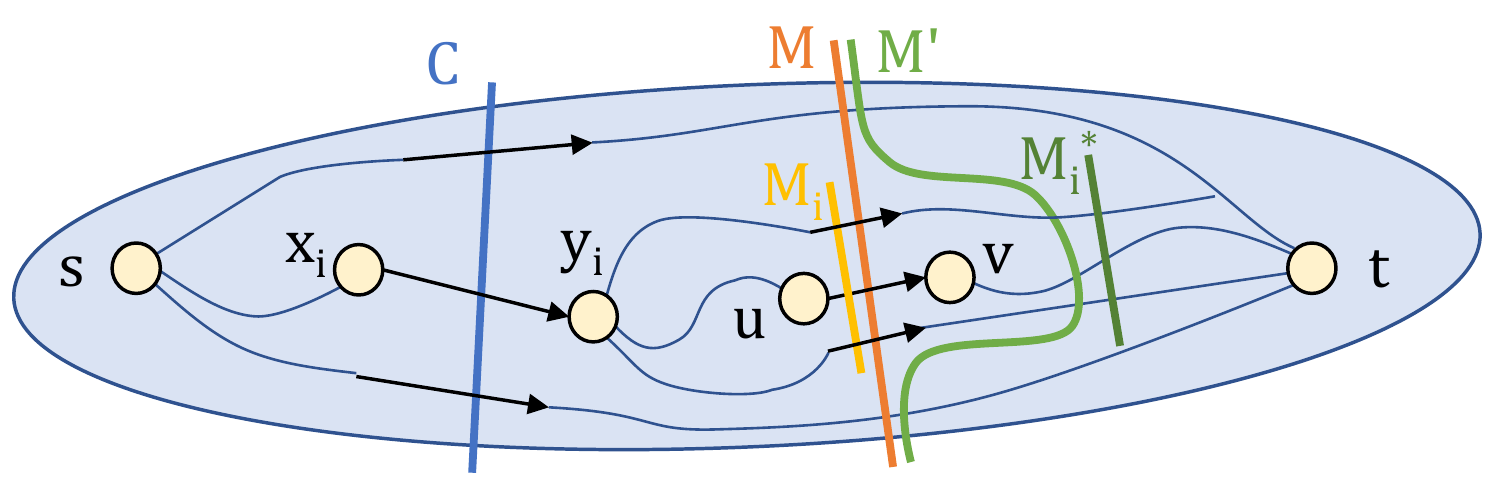}
\end{minipage}
\caption{(left) Illustration for the argument in Theorem~\ref{th:arcsplitcuts1}.
For the arc split $A_1$, $A_2$, the arcs in $A_1$ are shown in blue, the arcs in $A_2$ in black.
Bold arcs represent parallel arcs that cannot be cut.
$M$ is an $s$-$t$ min-cut, but it does not satisfy Property~\ref{property1} for vertex $v$.
Replacing $M_v^1$ by ${M'}_v^1$ gives $M'$ which does satisfy Property~\ref{property1}.
(right) Illustration for the argument in Theorem~\ref{th:arcsplitcuts2}.
$C$ is any cut that is not later than any cut in $\mathcal{F}_{s,t}$.
For any $s$-$t$-latest cut $M$ without Property~\ref{property2}, we can find a subset $M_i$, that is a $y_i$-$t$ cut and that can be replaced by some later $y_i$-$t$ cut $M_i^*$. This gives us a later $s$-$t$ cut $M'$ of equal or smaller size, contradicting $M$ being $s$-$t$ latest.
}
\label{fig:property_proofs}
\end{figure}

\begin{property}[Split-covering sets]
We say that a set $M \subseteq A$ is \emph{split-covering} with respect to arc split $A_1,A_2$ and vertices $s$,$t$ iff for any $v \in V$, there exists $M_v \subseteq M$ such that at least one of the following conditions holds
\begin{itemize}
\item $s\not=v\not=t$ and $v$ is either unreachable from $s$ in $(V,A_1)$ or $v$ does not reach $t$ in $(V,A_2)$. 
\item $M_v$ is an $s$-$v$-earliest cut in $(V,A_1)$,
\item $M_v$ is a $v$-$t$-latest cut in $(V,A_2)$.
\end{itemize}
\label{property1}
\end{property}

\begin{theorem}
\label{th:arcsplitcuts1}
Fix an arbitrary arc split $A_1,A_2$ and vertices $s$,$t$.
Any split-covering set (w.r.t. $A_1,A_2$ and $s$,$t$) is an $s$-$t$-cut in $G$. Moreover, there exists an $s$-$t$ min-cut that is split-covering (w.r.t. $A_1,A_2$ and $s$,$t$).
\end{theorem}
\begin{proof}
To argue that any set $M$ that is split-covering is an $s$-$t$-cut in $G$, consider any $s$-$t$ path $P$ in $G$.
We only consider vertices $v$ that are both reachable from $s$ in $(V,A_1)$ (or $v=s$) and reach $t$ in $(V,A_2)$ (or $v=t$).
Let $v$ be the last vertex in $P$ that is reached using arcs in $A_1$ ($v=s$ if $P$ has no edges from $A_1$).
Now the set $M_v \subseteq M$ is either an $s$-$v$-earliest cut (which implies $s \neq v$) and then $M_v$ intersects $P$ before $v$ or $M_v$ is a $v$-$t$-latest cut (which implies $v \neq t$) and then $M_v$ intersects $P$ after $v$.

To see that there is an $s$-$t$ min-cut that is split-covering, take any $s$-$t$ min-cut $M$.
Intuitively speaking, we now argue that we can incrementally push its arcs in $A_1$ towards $s$ and its arcs in $A_2$ towards $t$ until all its sub-cuts $M_v$ become earliest/latest. Throughout these changes $M$ is always an $s$-$t$ cut and never increases in size. We refer to Figure~\ref{fig:property_proofs} (left) for an illustration.

For any $v \in V$, we define $M_v^1$ as a minimum subset of $M$ such that $M_v^1$ is an $s$-$v$ cut in $(V, A_1)$ and $M_v^2$ as a minimum subset of $M$ such that $M_v^2$ is a $v$-$t$ cut in $(V, A_2)$.
Note that while only one of the two sets might exist, at least one does, otherwise $M$ would not be an $s$-$t$ cut.
Assume w.l.o.g. that $M_v^1$ exists. If $M_v^1$ is an earliest $s$-$v$ cut in $(V, A_1)$, vertex $v$ already satisfies Property~\ref{property1} and no change is required for $v$. Otherwise there is another cut ${M'}_v^1$ that is earlier than $M_v^1$ and satisfies $\lvert{M'}_v^1\rvert \leq \lvert{M}_v^1\rvert$. Now we define a new set $M' = (M \setminus M_v^1) \cup {M'}_v^1$ which satisfies $\lvert M' \rvert \leq \lvert M \rvert$.

We claim that $M'$ is still an $s$-$t$ min-cut in $G$.
For the sake of reaching a contradiction, assume otherwise. Any $s$-$t$ path $P$ that avoids $M'$ must use an arc $a$ in ${M}_v^1$. Let $P'$ denote the prefix of $P$ from $s$ to $a$. As $a \in A_1$, we have $P' \subseteq A_1$.
By definition of ${M}_v^1$, there is a path $Q$ from $a$ to $v$ in $A_1$ (otherwise $a$ would not need to be in ${M}_v^1$).
As ${M'}_v^1 \leq M_v^1$, we can pick $Q$ such that it avoids ${M'}_v^1$.
Thus concatenating $P'$ and $Q$ gives an $s$-$v$ path in $A_1 \setminus {M'}_v^1$, a contradiction.
Hence, $M'$ is also an $s$-$t$ cut and as $\lvert M' \rvert \leq \lvert M \rvert$ it is also an $s$-$t$ min-cut.
Applying this argument repeatedly for all $v$ without an earliest/latest $M_v$, will end with an $s$-$t$ min-cut that satisfies Property~\ref{property1} after finitely many repetitions.
\end{proof}

\begin{property}[Late-covering sets]
Let $C = \{(x_1,y_1),\ldots,(x_j,y_j)\}$ be any $s$-$t$ cut such that $M' \geq C$ for all $M' \in \mathcal{F}_{s,t}$.
We say that a set $M \subseteq A$ is \emph{late-covering} w.r.t. $C$ if for each $1 \le i \le j$, at least one of the following conditions holds
\begin{itemize}
\item $(x_i,y_i) \in M$,
\item there is $M_i \subseteq M$ such that $M_i$ is a $y_i$-$t$-latest cut.
\end{itemize}
\label{property2}
\end{property}

In order to show that there always exists a set $C$ satisfying Property \ref{property2} consider the set $C=\{(s,v): v \textit{ reaches } t \textit{ in } G\}$ which is an $s$-$t$ cut where for each $s$-$t$ latest cut $M$ it holds $M\geq C$.

\begin{theorem}
\label{th:arcsplitcuts2} Fix $s$-$t$ cut $C$ as in Property~\ref{property2}.
Any late-covering set w.r.t. $C$ is an $s$-$t$-cut. Every $s$-$t$-latest cut is late covering  w.r.t. $C$.
\end{theorem}
\begin{proof}
To argue that any late-covering set $M$ is an $s$-$t$-cut, consider any $s$-$t$ path $P$ in $G$.
As $C$ is an $s$-$t$ cut, $P$ must use some arc $(x_i, y_i)$ in $C$.
Since $M$ is late-covering, either $(x_i, y_i)$ or some $y_i$-$t$ cut is part of $M$, hence $M$ intersects $P$ either at $(x_i, y_i)$ or in some later arc.
Therefore, $P$ is not in $G \setminus M$, and since $P$ was arbitrary, $M$ is an $s$-$t$-cut in $G$.

To see that any $s$-$t$-latest cut is late-covering,
assume, for the sake of reaching a contradiction, that $M$ is an $s$-$t$-latest cut that is not late-covering.
As $M \neq C$ ($C$ is trivially late-covering) and $M$ is $s$-$t$-latest, we have $C < M$.
Hence $S_M \supseteq S_{C}$ and so $x_i \in S_M$ for all $i$.
For any $(x_i, y_i) \notin M$, $M$ must contain some $y_i$-$t$ cut $M_i$ (to ensure that $M$ is an $s$-$t$ cut).
As $M$ is not late-covering, there exists $i$ such that $M_i$ can not be chosen to be $y_i$-$t$-latest.
So for this $M_i$, there is another $y_i$-$t$ cut $M_i^*$ of size $|M_i^*| \leq |M_i|$ that is later than $M_i$.
We refer to Figure~\ref{fig:property_proofs} (right) for an illustration.

We now argue that $M' = (M \setminus M_i) \cup M_i^*$ is an $s$-$t$ cut that is later and not larger than $M$, which contradicts $M$ being $s$-$t$-latest.
Clearly, $|M'| \leq |M|$ as $|M_i^*| \leq |M_i|$.

We first argue that $M'$ is an $s$-$t$ cut. Take any $s$-$t$ path $P$ that avoids $M'$. $P$ must use an arc $a$ in $M_i$.
Let $P'$ denote the suffix of $P$ from $a$ to $t$. As $P'$ avoids $M'$, $P'$ is in $G \setminus M_i^*$.
As $M_i^* > M_i$, there is a $y_i$-$a$ path $Q$ in $G \setminus M_i^*$.
Concatenating $Q$ and $P'$ gives a $y_i$-$t$ path in $G \setminus M_i^*$, hence $P$ cannot exist.

It remains to argue that $M'$ is later than $M$, so $\lbar{T_{M'}} \subset \lbar{T_M}$.
Let us first argue that $\lbar{T_{M'}} \subseteq \lbar{T_M}$ by considering any $v \in \lbar{T_{M'}} \setminus \lbar{T_M}$ and reach a contradiction. Any $v$-$t$ path $P$ in $G \setminus M'$ must contain an arc $a \in M$ (as $v \notin \lbar{T_M}$) and we have $a \in M_i \setminus M_i^*$ (as $M_i \setminus M_i^*$ is where $M$ and $M'$ differ).
Let $P'$ be the suffix of $P$ from $a$ to $t$ and let $Q$ be any $y_i$-$a$ path in $G\setminus M_i^*$ (exists as $a \in M_i$ and $M_i^* > M_i$).
The concatenation of $Q$ and $P'$ would form a $y_i$-$t$ path in $G \setminus M_i^*$, a contradication.
Finally, $\lbar{T_{M'}} \neq \lbar{T_M}$, as for any $(u,v) \in M_i \setminus M_i^*$, we have $v \in \lbar{T_M}$ (otherwise $(u,v)$ would make $M$ non-minimal) and $v \notin \lbar{T_{M'}}$ (otherwise there would be a $y_i$-$t$ path in $G \setminus M_i^*$).
Hence $M' > M$ and $|M'| \leq |M|$, so $M$ is not $s$-$t$-latest.
\end{proof}

\begin{corollary}
\label{cor:atmostk}
In Theorem~\ref{th:arcsplitcuts1}, if $s$-$t$ min-cuts are of size at most $k$, then there is an $s$-$t$ min-cut that is split-covering in a way that every $M_v$ is either $s$-$v$-earliest $\leq$$k$-cut or $v$-$t$-latest $\leq$$k$-cut.
In Theorem~\ref{th:arcsplitcuts2}, every $s$-$t$-latest $\le k$-cut  is late covering in a way that every $M_i$ is a $y_i$-$t$-latest $\le k$-cut.
\end{corollary}

\section{Deterministic Algorithms with Witnesses for DAGs}
\label{sec:kapmc}

\subsection{\texorpdfstring{$k$-Bounded All-Pairs Min-Cut for $k=o(\sqrt{\log n}\,)$}{$k$-Bounded All-Pairs Min-Cut for k = o(sqrt(log n))}}
\label{sec:kapmc-dynamic-program}
We now develop a first algorithm for computing small cuts, which considers the vertices of the DAG one by one.
A first step is to consider a problem highlighted in Theorem~\ref{th:arcsplitcuts1}, namely how to pick a set of arcs that covers at least one earliest or one latest cut for every vertex. 
We formalize this problem, independently of graphs and cuts, as follows:
\begin{problem}[\witness]
\label{problem:sfc}
Given a collection of $c$ set families $\mathcal{F}_1,\mathcal{F}_2,\ldots,\mathcal{F}_c$, where each $\mathcal{F}_i$ is of size at most $K$, and each member of $\mathcal{F}_i$ is a subset of size at most $k$ of some universe~$U$, find all sets $W \subseteq U$ such that:
\begin{itemize}
\item \emph{[cover]} \label{en:prop1} for all $i$, there is $W_i \subseteq W$ such that $W_i \in \mathcal{F}_i$,
\item \emph{[size]} \label{en:prop2} $|W| \le k$,
\item \emph{[minimal]} \label{en:prop3} no proper subset of $W$ satisfies the [cover] condition.
\end{itemize}
\end{problem}

A naive solution to the \witness problem is to iterate through all $K^c$ possible unions of sets, one from each $\mathcal{F}_i$. However, using pruning as soon as the [size] constraint is violated, we can achieve a linear dependency on $c$ while keeping the exponential dependency on $k$.

\begin{algorithm}[t]
\caption{Solving Witness Superset by recursion with pruning}
\label{alg:pruning}
\SetKwFunction{KwRec}{SingleFamilyWitness}
\Def{$\KwRec( i, S)$}
{
	\uIf{$i = c+1$}
	{
		\uIf{no proper subset of $S$ satisfies the [cover] condition}
		{
			output $S$\;
		}	
		\Return\;
	}
	\uIf{$\exists F \in \mathcal{F}_i : F \subseteq S$}
	{
		$\KwRec(i+1,S)$\;
	}
	\uElse
	{
		\For{$F \in \mathcal{F}_i$}
		{
			\uIf{$|S \cup F| \le k$}
			{
				$\KwRec( i+1,S \cup F)$\;
			}
		}
	}
}\;
\Def{$\KwSet(\mathcal{F} = \{\mathcal{F}_1,\mathcal{F}_2,\ldots,\mathcal{F}_c\})$}
{
	
	$\KwRec(1,\emptyset)$\;
}
\end{algorithm}

\begin{lemma}
\label{lem:algoruntime}
Algorithm~\ref{alg:pruning} solves the \witness problem in time $\bigO(K^{k+1} \cdot c \cdot \textrm{poly}(k))$ and outputs a list of $\bigO(K^k)$ sets.
\end{lemma}
\begin{proof}
Whenever the function $\KwRec$ is called on input $i$, we can inductively argue that the candidate set $S$ is guaranteed to cover a set in each of the first $i-1$ families.
For the $i$-th family $\mathcal{F}_i$, $\KwRec$ adds new elements to $S$ only if necessary.
If $S$ already covers some $F \in \mathcal{F}_i$, $S$ remains unchanged.
Otherwise, we try all $F \in \mathcal{F}$ whose addition to $S$ do not break the [size] constraint separately by adding them to $S$ and evaluating recursively.

The correctness of Algorithm~\ref{alg:pruning} follows from the fact that this search skips only over those of the $K^c$ possible solutions that are either larger than $k$ or are not minimal, so those violating conditions [size] and [minimal] of the \witness problem.

To analyze the running time, we analyze the shape of the call tree $T$ of $\KwRec$.
Any root-to-leaf path in $T$ has length $c$ but only visits at most $k$ branching nodes, as each branching node increases the size of $S$ by at least one. Also, each branching node has out-degree at most $K$.
Hence $T$ is a tree of depth $c$ with at most $K^k$ leaves, which establishes our output size and implies $\lvert T \rvert \leq K^k \cdot c$.
The amount of work required to manipulate $S$ in each step is in $\bigO(\textrm{poly}(k))$.
Note that the final check before outputting $S$ can be done in $\bigO(K \cdot c \cdot \textrm{poly}(k))$ by just checking condition [cover] of the \witness problem for $S \setminus \{s\}$ for every $s \in S$ explicitly.
Combining $\lvert T \rvert \cdot \bigO(c \cdot \textrm{poly}(k))$ with $K^k \cdot \bigO(K \cdot c \cdot \textrm{poly}(k))$ gives our running time.
\end{proof}

\begin{theorem}
\label{th:dynprog}
All latest cuts of size at most $k$ for all pairs of vertices of a DAG can be found in $\bigO(2^{\bigO(k^2)}  \cdot m  n)$ total time.
\end{theorem}
\begin{proof}
We perform dynamic programming by repeatedly combining families of latest cuts while iterating through the vertices in a reverse topological order of $G$.
Let this order be denoted by $v_n, v_{n-1},\ldots , v_1$.
When processing vertex $v_i$ in this order, we compute all small cuts originating from $v_i$, so all $v_i$-$v_j$-latest $\leq$$k$-cuts for all $i \leq j$.
Note that for any $i > j$ all $v_i$-$v_j$ cuts are trivial since $G$ is a DAG.
To find the non-trivial cuts, consider $G' = (V',A')$, the subgraph of $G$ induced by vertices $V' = \{v_i,v_{i+1},\ldots,v_n\}$.
Since $G$ is a DAG, cuts between those vertices are preserved in $G'$.
Consider the arc split of $G'$ with $A_1$ being all arcs leaving $v_i$ and $A_2 = A' \setminus A_1$.

For every $j$ such that $i<j\le n$, we build an instance $\mathcal{I}_{i,j}$ of the set families cover the \witness problem to find all $v_i$-$v_j$-latest $\leq$$k$-cuts.
We set $c$ equal to the out-degree of $v_i$ and let $N^+(v_i) = \{v_{i_1}, \ldots, v_{i_c}\}$ denote all the heads of the arcs leaving $v_i$.
Then, for every $x$ such that $1 \leq x \leq c$, the set family $\mathcal{F}_x$ is composed of all the cuts in $\mathcal{F}_{v_{i_x},v_j}^{\le k}$ and the singleton $\{(v_i,v_{i_x})\}$.

Note that $\{(v_i,v_{i_x})\}$ is the only earliest $v_i$-$v_{i_x}$ cut in $(V', A_1)$.
Furthermore, for all other vertices $v \in V' \setminus N^+(v_i)$, the empty set is the only earliest $v_i$-$v$ cut in $(V', A_1)$.
Since all arcs leaving $v_i$ are in $A_1$, the empty set is the only $v_i$-$t$-latest cut in $(V', A_2)$.
Hence the families $\mathcal{F}_1, \mathcal{F}_2, \dots, \mathcal{F}_c$ contain all earliest cuts in $(V', A_1)$ and all latest cuts in $(V', A_2)$ for all of $V'$.
Therefore, we can apply Theorem~\ref{th:arcsplitcuts1} and any solution to $\mathcal{I}_{i,j}$ corresponds to a $v_i$-$v_j$-cut with Property~\ref{property1}.

As $(S_{A_1}, \lbar{S_{A_1}}) = (\{v_i\}, V'\setminus(\{v_i\}))$ is the earliest possible $v_i$-$v_j$ cut, we can apply Theorem~\ref{th:arcsplitcuts2} and Corollary~\ref{cor:atmostk} as well, using $C = (S_{A_1}, \lbar{S_{A_1}})$.
Notice that by Theorem~\ref{th:arcsplitcuts2} every $v_i$-$v_j$-latest cut is late covering w.r.t. $C$, and as the solution to the \witness problem finds all late covering cuts we know that all $v_i$-$v_j$-latest $\le$$k$-cuts appear among the solutions to $\mathcal{I}_{i,j}$.

Since all solutions to $\mathcal{I}_{i,j}$ are guaranteed to be minimal,
all non-minimal $v_i$-$v_j$-cuts with properties~\ref{property1}~and/or~\ref{property2} get filtered out and only the minimal $v_i$-$v_j$ cuts remain.
But some of the $v_i$-$v_j$ cuts among the solutions to $\mathcal{I}_{i,j}$ might not be $v_i$-$v_j$-\emph{latest}.

To enumerate the solutions to $\mathcal{I}_{i,j}$ we call Algorithm~\ref{alg:pruning}, which takes $\bigO(2^{\bigO(k^2)} \cdot c)$ time, by Lemma~\ref{lem:algoruntime} and the $K \in 2^{\bigO(k)}$ bound on the number of latest $\le$$k$-cuts from Lemma~\ref{lem:catalan}.
This time bound, summed over all source vertices and all target vertices gives the claimed total runtime $\bigO(2^{\bigO(k^2)}  \cdot m  n)$.

To complete the proof, we need to argue that we can filter the solutions to $\mathcal{I}_{i,j}$ resulting only in the $v_i$-$v_j$-latest cuts, without additionally changing the time complexity.
Let $\mathcal{M}$ be the output of Algorithm~\ref{alg:pruning}, so all the $v_i$-$v_j$ cuts satisfy Property~\ref{property2}.
We show inductively that we can compute the later-relation for the cuts within $\mathcal{F}^{\le k}_{v_i,v_j} \subseteq \mathcal{M}$ while doing the filtering.

We explicitly compute the relative order for all the $v_i$-$v_j$ cuts in $\mathcal{M}$.
That is, we test all pairs $M,M' \in \mathcal{M}$, for the following relation:
\begin{align}
\label{eq:fx-order}
M' \ge M \Leftrightarrow \text{ $\forall x \in [c]$, $ \forall M_x,M'_x \in \mathcal{F}_x:$ $M_x \subseteq M$ and $M'_x \subseteq M'$ implies $M'_x \ge M_x$.}
\end{align}

For this order to be well-defined, we need the later-relation $M'_x \ge M_x$ to be defined for all the sets within $\mathcal{F}_x$.
This is not immediate since $\mathcal{F}_x$ also contains the singleton set $\{(v_i,v_{i_x})\}$ which might not disconnect $v_i$ from $v_j$.
To fix this, we just define that $\{(v_i,v_{i_x})\}$ shall be considered earlier than all cuts $M \in\mathcal{F}_{v_{i_x},v_j}^{\le k}$.
This suffices to extend the order to all of $\mathcal{F}_x$ as we already know the relative order for all pairs of $v_{i_x}$-$v_j$-latest $\le$$k$-cuts, by the inductive assumption.

To argue about the correctness of the equivalence in \eqref{eq:fx-order}, observe that for every $x$, and $M_x, M'_x \in \mathcal{F}_x$, $M'_x \ge M_x$ holds if and only if every path from $v_i$ to $v_j$ using $(v_i,v_{i_x})$ as its first arc intersects $M_x$ no later than intersecting $M'_x$.
Thus, this partial later-order on $\mathcal{F}_x$ properly extends to the partial later-order that we want on (minimal) cuts, and to filter $\mathcal{M}$ it is enough to keep those $M \in \mathcal{M}$ such that for no $M' \in \mathcal{M}$ there is $M' > M$ and $|M'| \le |M|$.

The total amount of work for the filtering for a single pair of $i,j$ is thus quadratic in $\lvert \mathcal{M} \rvert = 2^{\bigO(k^2)}$, the number of cuts in $\mathcal{M}$, linear in $c$, the number of set families $\mathcal{F}_x$, and quadratic in $2^{\bigO(k)}$, the size of each set family. As $\textrm{poly}(2^{\bigO(k^2)}) \cdot c \cdot \textrm{poly}(2^{\bigO(k)}) = 2^{\bigO(k^2)} \cdot c$, where $c$ is the outdegree of $v_i$, we get the claimed bound on the runtime.
\end{proof}

\begin{corollary}
All latest cuts of size at most $k = o(\sqrt{\log n}\,)$ for all pairs of vertices of a DAG can be found in $\bigO(m  n^{1+o(1)})$ total time.
\end{corollary} 
\subsection{Coding for the Witness Superset Problem}

\subparagraph{Binary codes.}
A binary code $\mathcal{C}$ of length $q$ on a universe of size $u$ is a function $[u] \to 2^{[q]}$.
Note that we employ set formalism to describe codes, i.e., each element of the universe is mapped to a subset of the code set.
This is equivalent to mapping to binary codes of length $q$, used for instance in implementations.
Correspondingly, we talk about the bitwise-OR operation on codewords, which is equivalent under the set formalism to taking the union of two characteristic sets.

\subparagraph{Tensor products and powers.}
Given two binary codes $\mathcal{C}_1$, $\mathcal{C}_2$, we define the \emph{tensor product} $\mathcal{C}_1 \otimes \mathcal{C}_2: [u_1]\times[u_2] \to 2^{[q_1] \times [q_2]}$ of the two codes $\mathcal{C}_1$ and $\mathcal{C}_2$ as the function $(\mathcal{C}_1 \otimes \mathcal{C}_2) (w_1,w_2) = \mathcal{C}_1(w_1) \times \mathcal{C}_2(w_2)$.\footnote{$\mathcal{C}_1 \times \mathcal{C}_2$ is interpreted as a code from natural bijection between $[u_1]\times [u_2]$ and $[u_1 \cdot u_2]$.}
The tensor product resembles the construction of concatenated codes: instead of the outer code $\mathcal{C}_1(w_1)$, each '1' in $\mathcal{C}_1(w_1)$ is replaced with $\mathcal{C}_2(w_2)$, and each `0' with a sequence of '0' of appropriate length.
We call the replacement of each entry of $\mathcal{C}_1(w_1)$ (that is, by $\mathcal{C}_2(w_2)$ or by '0's) a \emph{column} of the codeword $(\mathcal{C}_1 \otimes \mathcal{C}_2) (w_1,w_2)$.

We define the $p$-th \emph{tensor power} $\mathcal{C}^{\otimes p}: [u]^p \to 2^{[q]^p}$ of a code $\mathcal{C}$ as the tensor product of $p$ copies of $\mathcal{C}$:
$$\mathcal{C}^{\otimes p} = \underbrace{\mathcal{C} \otimes \ldots \otimes \mathcal{C}}_{p \text{ times}},$$
which is equivalent to taking the tensor product after applying the code to each argument:
$$(\mathcal{C}^{\otimes p})(w_1,w_2,\ldots,w_p) = \mathcal{C}(w_1) \times \mathcal{C}(w_2) \times \ldots \times \mathcal{C}(w_p).$$

\subparagraph{Superimposed codes.}
To apply the code to a subset $X$ of $[u]$, we write $\mathcal{C}(X) = \bigcup_{x \in X} \mathcal{C}(x)$.
A binary code $\mathcal{C}$ is called $d$-superimposed \cite{mooers1948application}, if the union of at most $d$ codewords is uniquely decodable, or equivalently
$$\forall_{X : |X| \le d} \forall_{y \not\in X} \mathcal{C}(y) \not\subseteq \mathcal{C}(X).$$
We refer to Figure~\ref{fig:superimposed_codes} for an illustration.

\begin{figure}
\begin{minipage}{\textwidth}
\centering\includegraphics[trim={0 9cm 0 0},width=\textwidth]{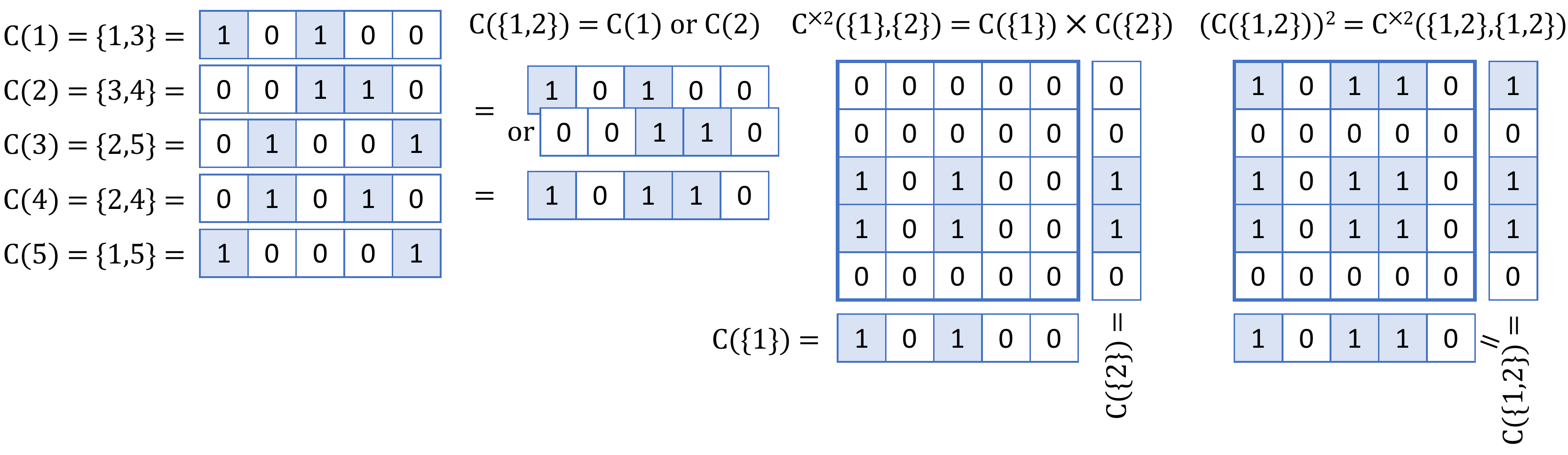}
\end{minipage}
\caption{(left) A $2$-superimposed code $\mathcal{C}$ on the universe $[5] =\{1,2,3,4,5\}$.
Notice how the encoding is unique for any set of size at most $2$.
For instance, the code $\mathcal{C}(\{1,2\})$ shown is uniquely decodable to $\{1,2\}$.
This is not true for larger sets however, e.g. $\mathcal{C}(\{1,2,3\}) = \mathcal{C}(\{2,3,5\})$.
(right) The second tensor power $\mathcal{C}^{\times 2}$ of the code $\mathcal{C}$ illustrated as a two dimensional matrix and applied once to the two singleton sets $\{1\}$ and $\{2\}$, and once to the set $\{1,2\}$ in both dimensions. 
\label{fig:superimposed_codes}}
\end{figure}

A standard deterministic construction based on Reed-Solomon error correction codes (cf.~\cite{ReedSolomon}) gives  $d$-superimposed codes of length $d \log^2 u$ (cf. Kautz and Singleton~\cite{KautzSingleton}). This construction is very close to the information-theoretic lowerbound of $\Omega(d \log_d u)$. By decoding a superimposed code, we understand computing $X=\mathcal{C}^{-1}(Y)$, if $|X|\le d$, or deciding that there is no such $X$. When considering decoding time, Indyk et al.~\cite{IndykNgoRudra} gave the first (randomized) construction of superimposed codes decodable in time $\bigO(\textrm{poly}(d \log u))$ that is close to the lower-bound on length of codes, while Ngo et al.~\cite{Ngo2011} provided a derandomized construction. However, for our purposes any superimposed code of length and decoding time $\bigO(\textrm{poly}(d \log u))$ is sufficient, thus we use the following folklore result (c.f. ''bit tester'' matrix \cite{DBLP:conf/stoc/GilbertSTV07}).
\begin{theorem}
\label{th:fast_superimposed}
Let $\mathcal{C}_{\text{slow}}$ be an arbitrary $d$-superimposed code of length $\bigO(\textrm{poly}(d \log u))$ (for example by the Kautz-Singleton construction), and let $\mathcal{C}_1$ be a $1$-superimposed code decodable in $\bigO(\log u)$ time and space (for example the code adding a parity bit after every bit of the input's binary representation). Then $\mathcal{C}_\text{fast} = \mathcal{C}_\text{slow} \otimes \mathcal{C}_1$ is a $d$-superimposed code decodable in $\bigO(\text{poly}(d \log u))$ time and space.
\end{theorem}
\begin{proof}
$\mathcal{C}_{\text{fast}}$ being $d$-superimposed follows from it having non-zero column only in the positions corresponding to non-zero bits of $\mathcal{C}_{\text{slow}}$.
Since $\mathcal{C}_\text{slow}$ is $d$-superimposed, for any $X$ such that $|X| \le d$, for all elements $x$ of $X$, there exists a coordinate that is in the encoding of $X$ only because of $x$, i.e., there is exists $i$ such that $i\in \mathcal{C}_\text{slow}(x)$ and $i \not\in \mathcal{C}_\text{slow}(X \setminus \{x\})$.
Thus, given any $S \subseteq [q_{\text{slow}}]\times [q_1]$ to be decoded by $\mathcal{C}_\text{fast}$, it is enough to first find the set $\mathcal{I}$ of all $i$ such that $i$-th column of $S$ is a proper code-word of $\mathcal{C}_1$, and test if $\mathcal{C}_\text{slow}(\mathcal{I}) \times \mathcal{C}_1(\mathcal{I}) = S$.
\end{proof}
One can see superimposed codes as a short encoding of sets, that are decodable up to a certain size, and preserve set union under set codeword bitwise-OR.
Now we develop codes that allow for encoding of set families, with the goal of decoding solutions to \witness from results of the bitwise-OR of encoded inputs.

More specifically, assume $\mathcal{C}$ is a $k$-superimposed code (on a universe of size $u$), and consider the code $\mathcal{C}^{\otimes K}$. If $\mathcal{F} = \{F_1,F_2,\ldots,F_{K}\}$ is a set family of $K$ sets of at most $k$ elements from $[u]$, then by slightly bending the notation there is $$(\mathcal{C}^{\otimes K})(\mathcal{F}) = \mathcal{C}(F_1) \times \ldots \times \mathcal{C}(F_K),$$
(the order of sets $F_1,\ldots,F_K$ is chosen arbitrarily). If $|\mathcal{F}| < K$ then we append several copies of, i.e., $F_1$ in the encoding. Moreover, we need to guarantee that $\mathcal{F} \not= \emptyset$.

The intuition behind this construction is to assign a separate dimension for each set of the set family. This creates enough redundancy so that the code has some desired properties under bitwise-OR. Simpler codes, for example concatenating instead of taking a tensor product, do not have those properties.

\subparagraph{Taking slices.}
We define taking \emph{slices} of smaller dimension from multidimensional sets, by fixing one (or more) coordinates to specific values, that is taking a subset of the original set that has fixed coordinates equal to the desired ones, and then eliminating those coordinates from the tuples. For example, for a set
$\{(0,1,2),(2,2,0),(3,1,2)\}$
setting its second coordinate to $'1'$ results in the slice
$\{(0,2),(3,2)\}.$ This is a simple operation that reduces the level of redundancy in the code: observe that a slice of $\mathcal{C}^{\otimes K}(\mathcal{F})$ is just an encoding $\mathcal{C}^{\otimes K'}(\mathcal{F}')$ for some $\mathcal{F}' \subset \mathcal{F}$ and $K' < K$.

\subparagraph{Solving \witness.}
We now develop Algorithm~\ref{alg:boxes}, which takes an encoded input $S$ of an instance of \witness and solves it by recursively taking unions $S'$ of at most $k$ different slices of $S$.
This way, each recursive call reduces the dimension by one from $\mathcal{C}^{\otimes K}$ to $\mathcal{C}^{\otimes K-1}$.
Correctness follows from formalizing the following property: for each solution $W$, we are sure that (at least) one recursive call considers a slice-union $S'$ for which each element $x \in W$ is either easily observed already in $S$ or we are sure that $x$ is still necessary within $S'$.
We refer to Figure~\ref{fig:boxalgorithm} for an illustration.

\begin{figure}
\begin{minipage}{\textwidth}
\centering\includegraphics[scale=0.52]{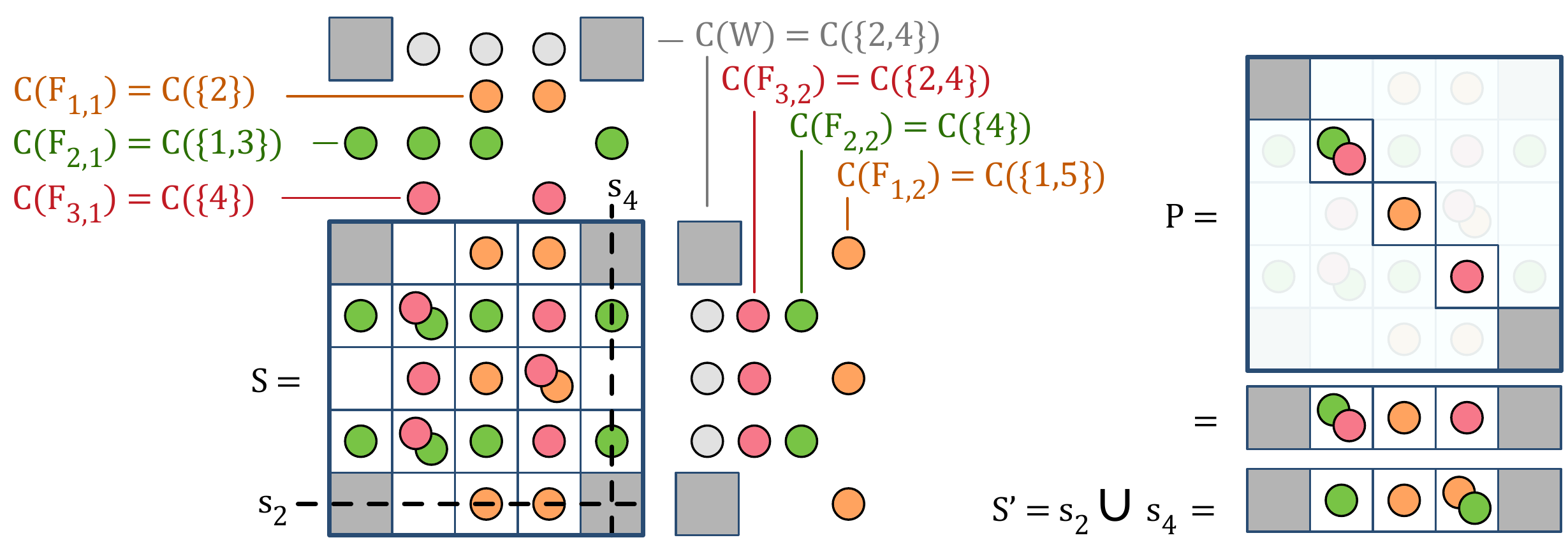}
\end{minipage}
\caption{Illustration of Algorithm~\ref{alg:boxes} solving an instance of \witness, namely the set families $\mathcal{F}_1 = \{\{2\},\{1,5\}\}$, $\mathcal{F}_2 = \{\{1,3\},\{4\}\}$ and $\mathcal{F}_3 = \{\{4\},\{2,4\}\}$, so we have $u=5$, $c=3$, $K=2$ and $k=2$.
We use the same superimposed code $\mathcal{C}$ of length $t=5$ as in Figure~\ref{fig:superimposed_codes} to form the input $S$.
$S$ is the bitwise-OR encoding of the set families, namely
{$S = (\mathcal{C}(F_{1,1})\times \mathcal{C}(F_{1,2})) \cup (\mathcal{C}(F_{2,1})\times \mathcal{C}(F_{2,2})) \cup (\mathcal{C}(F_{3,1})\times \mathcal{C}(F_{3,2}))$}.
The only solution to this instance is $W = \{2,4\}$.
As shown by Lemma~\ref{lem:propertyofsolution}, we have $S \cap ([q] \setminus \mathcal{C}(W))^2 = \emptyset$, i.e., all gray boxes are free from colored dots.
Note that both cases of the argument in Theorem~\ref{th:codes_decoding} apply here.
To see that the first case applies for $x=4$, note that $\mathcal{F}_3$ is $4$-$W$-critical, and that we have $F_{3,1} \subseteq W$ and $F_{3,2} \subseteq W$.
Thus, $\mathcal{C}(4) = \{2,4\} \subseteq P$.
The set $P$ corresponds to the diagonal of the matrix $S$.
In this specific example, we even have $\mathcal{C}(W) \subseteq P$ because $P = \{r \mid (r,r) \in S\} = \{2,3,4\}$ contains element $3$ as well since $(3,3) \in \mathcal{C}(\mathcal{F}_1) \subseteq S$.
For the second case, the two slices $s_2$ and $s_4$ are shown as dotted lines.
For $x=2$, $\mathcal{F}_1$ is $x$-$W$-critical and $F_{1,2} \not\subseteq W$.
Hence, we can still observe $2$ within slice $s_2$.
Analogously for $x=4$, $\mathcal{F}_2$ is $x$-$W$-critical and $F_{2,1} \not\subseteq W$ and thus $s_4$ still enforces that $W \setminus \{4\}$ is not a solution.
In the recursion for $S' = s_2 \cup s_4$, we get $S' = \mathcal{C}(W)$, which corresponds to the base case of the induction.
\label{fig:boxalgorithm}}
\end{figure}

\begin{algorithm}[h!]
\caption{decoding of $\mathcal{C}^{\otimes K}$}
\label{alg:boxes}
\SetKwFunction{KwDecode}{decodeWitness}
\KwIn{$S$, bitwise-OR of set families encoded with $\mathcal{C}^{\otimes K}$}
\KwOut{$\KwDecode(S)$ outputs all possible solutions to the \witness problem on $S$ }
\SetKwFunction{KwCollapse}{collapse}
\Def{$\KwCollapse(S)$}
{
	$ans \gets \emptyset$\;
	\uIf{$K = 1$}
	{
		$ans.insert(S)$\;
	}
	\uElse
	{
		$P \gets \{ r\ |\ (r,r,\ldots,r) \in S\}$\;
		\Comment{$P$ are those coordinates that have to be in any solution}
		\For{$s_1, s_2, \ldots, s_k \in$ all $(K-1)$-dimensional slices of $S$}
		{
			$S' \gets \bigcup_{i} s_i$\;
			\For{ $\mathcal{I} \in \KwCollapse(S')$ }
			{
				$ans.insert(\mathcal{I} \cup P)$\;
			}
		}
	}
	\Return $ans$\;
}\;
\Def{$\KwDecode(S)$}
{
	$solutions \gets \emptyset$\;
	$q \gets$ length of $\mathcal{C}$\;
	\For{$\mathcal{I} \in \KwCollapse(S)$}
	{
		\uIf{$\mathcal{I}$ is decodable by $\mathcal{C}$ into a set of size at most $k$}
		{
			\uIf{$S \cap \Big([q] \setminus \mathcal{I}\Big)^K = \emptyset$  \label{lst:check1}}
			{
				$W \gets \mathcal{C}^{-1}(\mathcal{I})$\;
				\uIf{ $\forall\ W' \subset W$: $S \cap \Big([q] \setminus \mathcal{C}(W')\Big)^K \not= \emptyset$ \label{lst:check2} }
				{
					\Comment{$W$ is minimal}
					$solutions.insert(W)$\;
				}
			}
		}
	}
	\Return $solutions$\;
}
\end{algorithm}

\begin{theorem}
\label{th:codes_decoding}
Algorithm~\ref{alg:boxes} solves the \witness problem from the bitwise-OR of
$$(\mathcal{C}^{\otimes K})(\mathcal{F}_1), \ldots, (\mathcal{C}^{\otimes K})(\mathcal{F}_c)$$
in time $\bigO((K k \log u)^{\bigO(K \cdot k)})$, assuming all set families $\mathcal{F}_i$ are over the universe $[u]$, and $\mathcal{C}$ is of size $\bigO(poly(K \log u))$ and has a fast decoding procedure (i.e.,~the construction from Theorem~\ref{th:fast_superimposed}).
\end{theorem}
\begin{proof}
Let $q = \bigO(\textrm{poly}(k \log u))$ be the length of codewords in $\mathcal{C}$.
Denote $S \subseteq [q]^K$ as the input to be decoded, and denote the (to us unknown) sets from the set families of the input as: $\mathcal{F}_i = \{F_{i,1}, F_{i,2}, \ldots, F_{i,K}\}$.
First, we show that there is a characterization of solutions in the language of tensor products.
\begin{lemma}
\label{lem:propertyofsolution}
For any $W \subseteq [u]$: $W$ satisfies [cover] if and only if $S \cap ([q] \setminus \mathcal{C}(W))^K = \emptyset$.
\end{lemma}
\begin{proof}
We start with the definition of [cover] (1) and the fact that $\mathcal{C}$ properly represents set containment for sets of size at most $k$ (2):
\begin{align*}
\text{$W$ satisfies [cover]}
\stackrel{(1)}{\Leftrightarrow} \text{$\forall i$ $\exists j$: $F_{i,j} \subseteq W$ }
\stackrel{(2)}{\Leftrightarrow} \text{ $\forall i$ $\exists j$: $\mathcal{C}(F_{i,j}) \subseteq \mathcal{C}(W)$. }
\end{align*}

We use $\mathcal{C}(\mathcal{F}_{i}) = \mathcal{C}(F_{i,1}) \times \dots \times \mathcal{C}(F_{i,j})\times \dots \times \mathcal{C}(F_{i,K})$ and $F_{i,j'} \subseteq [q]$ for all $j' \neq j$ (3) and the equivalence $A \subseteq B \Leftrightarrow A \cap ([q] \setminus B) = \emptyset$ (4) to get
\begin{align*}
\text{ $\forall i$ $\exists j$: $\mathcal{C}(F_{i,j}) \subseteq \mathcal{C}(W)$ }
\stackrel{(3)}{\Leftrightarrow}& \text{ $\forall i$ $\exists j$: $\mathcal{C}(\mathcal{F}_{i}) \subseteq [q]^{j-1} \times \mathcal{C}(W) \times [q]^{K-j}$}\\
\stackrel{(4)}{\Leftrightarrow}& \text{ $\forall i$ $\exists j$: $\mathcal{C}(\mathcal{F}_{i}) \cap ([q]^{j-1} \times ([q] \setminus\mathcal{C}(W)) \times [q]^{K-j}) = \emptyset$.}
\end{align*}

We now exploit that the codes $\mathcal{C}(\mathcal{F}_{i})$ are tensor products to get the following equivalence
\begin{align*}
\text{ $\forall i$ $\exists j$: $\mathcal{C}(\mathcal{F}_{i}) \cap ([q]^{j-1} \times ([q] \setminus\mathcal{C}(W)) \times [q]^{K-j}) = \emptyset$}
\stackrel{(5)}{\Leftrightarrow}& \text{ $\forall i$: $\mathcal{C}(\mathcal{F}_{i}) \cap ([q] \setminus\mathcal{C}(W))^K = \emptyset$}\\
\stackrel{(6)}{\Leftrightarrow} \text{ $(\bigcup_i\mathcal{C}(\mathcal{F}_{i})) \cap ([q] \setminus\mathcal{C}(W))^K = \emptyset$ }
\stackrel{(7)}{\Leftrightarrow}& \text{ $S \cap ([q] \setminus\mathcal{C}(W))^K = \emptyset$}.
\end{align*}
Taking the union (6) and applying the definition of S (7) concludes the proof.
\end{proof}

Now, consider any set $W$ that we output. Obviously $|W| \le k$, and applying Lemma~\ref{lem:propertyofsolution} to the checks on lines~\eqref{lst:check1} and~\eqref{lst:check2} of Algorithm~\ref{alg:boxes} ensures that $W$ satisfies [cover] and [minimal]. Thus, $W$ is a solution to the \witness problem.

Next, we reason that any solution $W$ to the \witness problem
is generated by the $\KwCollapse$ function, by induction on $K$. If $K=1$, then $S$ is the bitwise-OR (union) of the encoded sets, and since $\KwCollapse$ outputs $S$, the condition is trivially satisfied.

For the inductive step, we first observe that any $(K-1)$-dimensional slice $s_i$ of $S$ is the bitwise-OR of $\mathcal{C}^{\otimes K-1}$ encoded set families, and thus so is $S'$.

Consider any slice $s$ of $S$ by fixing a single dimension to a value $\alpha \not\in \mathcal{C}(W)$.
Then, similar to the arguments in Lemma~\ref{lem:propertyofsolution}, we have $s \cap ([q] \setminus \mathcal{C}(W) )^{K-1} = \emptyset$. Additionally, observe that $P \subseteq \mathcal{C}(W)$ as otherwise $([q] \setminus \mathcal{C}(W))^K$ would intersect $S$.

Now, fix any $x \in W$. Since $W$ is inclusion minimal, there is at least one family $\mathcal{F}_i$ that requires $x$ to be in $W$, which we call $x$-$W$-critical.
\begin{property}[critical family]
$\mathcal{F}_i$ is $x$-$W$-critical, iff for all $j$, if $F_{i,j} \subseteq W$ then $x \in F_{i,j}$.
\label{prop:critical-family}
\end{property}
	
We now consider two, not necessarily disjoint cases based on all those families in $\mathcal{F}_1, \mathcal{F}_2, \dots, \mathcal{F}_c$ that are $x$-$W$-critical (see Figure~\ref{fig:boxalgorithm} for an illustration):
\begin{itemize}
\item There is a $x$-$W$-critical family $\mathcal{F}_i$ such that for all $j$, $F_{i,j} \subseteq W$. By Property~\ref{prop:critical-family} we have for all $j$, $x \in F_{i,j}$, and thus $(\mathcal{C}(x))^K \subseteq \mathcal{C}(\mathcal{F}_i) \subseteq S$ and hence $\mathcal{C}(x) \subseteq P$.
\item There is at least one $x$-$W$-critical family $\mathcal{F}_i$, such that for at least one value of $j'$ we have $F_{i,j'} \not\subseteq W$.
Hence $\mathcal{F}_{i}\setminus \{F_{i,j'}\}$ is also $x$-$W$-critical.
Now consider the slice $s_x$ of $S$ by fixing its $j'$-th dimension to some arbitrarily chosen $\alpha \in \mathcal{C}(F_{i,j'})\setminus \mathcal{C}(W)$.
Since slice $s_x$ is a hyperplane going through $\mathcal{C}^{\otimes K}(\mathcal{F}_i)$, $s_x$ contains $\mathcal{C}(\mathcal{F}_i \setminus \{F_{i,j'}\})$, and hence $W \setminus \{x\}$ is not a solution for $s_x$, but $W$ is.
\end{itemize}
Thus, we reach the conclusion that for any $x \in W$, there either exists a slice $s_x$ of $S$ that requires $x$ in at least one of its minimal solutions, and has $W$ as a solution, or $x$ is encoded in $P$. Let $W^* = \{x \in W \mid s_x \text{ exists}\}$. Thus, there is $S' = \bigcup_{x\in W^*} s_x$ such that $S'$ has a minimal solution $W'$ and such that $\mathcal{C}(W') \cup P = \mathcal{C}(W)$.
Since Algorithm~\ref{alg:boxes} exhaustively searches all combinations of at most $k$ slices deterministically, $S'$ is used in at least one of the recursive calls of $\KwCollapse$.

To bound the running time, observe that the number of slices on a single level of the recursion is $K \cdot t$, thus the branching factor of $\KwCollapse$ is upper-bounded by $(K t)^{k}$, with recursion depth $K$. All codes are of size $q \in \bigO(t^K)$, and computing $\mathcal{C}$ and $\mathcal{C}^{-1}$ takes time $\mathrm{poly}(t)$. All in all this leads to the claimed time $\bigO((K k \log u)^{\bigO(K k)})$.
\end{proof}
\subsection{\texorpdfstring{$k$-Bounded All-Pairs Min-Cut for $k=o(\log\log n)$}{$k$-Bounded All-Pairs Min-Cut for k=o(log log n)}}

\begin{theorem} \label{thm:latestDense} 
All latest cuts of size at most $k$ for all pairs of vertices of a DAG can be found in $\bigO((k \log n)^{4^{k+o(k)}} \cdot n^{\omega})$ total time.
\end{theorem}
\begin{proof}
We show a divide-and-conquer algorithm. Without loss of generality, assume that $n$ is even (if not, add one unique isolated vertex). Let $V_1 = v_1,v_2,\ldots,v_{n/2}$ and $V_2 = v_{n/2+1}, \ldots, v_{n-1}, v_{n}$, and let $A_1 = A(G[V_1])$, $A_{1,2} = A[V_1,V_2]$ and $A_2 = A(G[V_2])$. It is enough to show how to find all pairs earliest/latest $\le$$k$-cuts in $G$, having recursively computed all earliest/latest $\le$$k$-cuts in $(V_1,A_1)$ and in $(V_2,A_2)$ in the claimed time bound, since the recursive equation for the runtime $T(n) = 2T(n/2) +  \bigO((k\log n)^{4^{k+o(k)}} \cdot n^{\omega})$ has the desired solution.

Notice that the pairwise cuts between vertices in $V_1$ are the same in $G$ as in $(V_1,A_1)$, and the same holds for $V_2$ and $(V_2,A_2)$.
Thus, all we need is to find pairwise cuts from $V_1$ to $V_2$ in $G$. We proceed as follows.
First of all, we merge the information on pairwise cuts in $(V_2,A_2)$ with arcs from $A_{1,2}$ to compute all pairwise cuts from $V_1$ to $V_2$ in the graph $(V,A_{1,2} \cup A_2)$.
The second step is to merge the result of the first step with all pairwise cuts in $(V_1,A_1)$ to compute the desired pairwise cuts in $G$.
Since both procedures involve essentially the same steps, we describe in detail only the first one.

\textit{Encoding:}
Let $K \le 4^k$ be the bound from Lemma~\ref{lem:catalan}.
We use $\mathcal{C}^{\otimes K}$ described in Theorem~\ref{th:codes_decoding} with a universe of size $m$ to represent pairwise latest/earliest cuts\footnote{For multigraphs, we require $m \le 2^{\text{polylog}(n)}$ for our bound to hold.}.
We then build two encoded matrices of dimension $n/2 \times n/2$.
Matrix $Y$, defined by $Y_{i,j} = \mathcal{C}^{\otimes K}(\mathcal{E}^{\le k}_{v_i,v_j})$, encodes all earliest $v_i$-$v_j$ $\le$$k$-cuts in the graph $(V_2,A_2)$.
Matrix $X$ encodes cuts from $V_1$ to $V_2$ in graph $(V,A_{1,2})$, which has a much simpler structure: the cut is the set of all arcs from $v_i$ to $v_j$, or there is no cut if there are more than $k$ parallel arcs.

\textit{Matrix multiplication:} The following procedure is used
\begin{enumerate}
\item Lift $X$ into matrix $X'$, changing each entry from a $K$-dimensional tensor product to $2K$-dimensional, by setting $X'_{i,j} = X_{i,j} \times [t]^K$.
\item Lift $Y$ into $Y'$ by setting $Y'_{i,j} = [t]^K \times Y_{i,j}$.
\item Compute the coordinate-wise Boolean matrix product of $X'$ and $Y'$, resulting in $Z'$.
\end{enumerate}
We denote the above steps as $X \star Y = Z'$.
We note that if for some indices $a$, $b$, $c$, the entry $X_{a,b}$ encodes some set family $\mathcal{E}$ and $Y_{b,c}$ encodes some set family $\mathcal{F}$, then the bitwise product of $X'_{a,b}$ and $Y'_{b,c}$ encodes the set family $\mathcal{E}\cup\mathcal{F}$.
Thus, every entry $Z'_{i,j}$ is a bitwise-OR of all $v_i$-$v_a$-earliest $\le$$k$-cuts in $(V,A_{1,2})$ and all $v_a$-$v_{j+n/2}$-latest $\le$$k$-cuts encoded with $\mathcal{C}^{\otimes 2K}$.
Hence, by Theorem~\ref{th:codes_decoding}, Algorithm~\ref{alg:boxes} applied to each entry of $Z'$ solves \witness.
Since $A_{1,2}, A_2$ is an arc split of $(V,A_{1,2} \cup A_2)$, by Theorem~\ref{th:arcsplitcuts1} each solution in the output is a $v_i$-$v_{j+n/2}$ cut in $(V,A_{1,2} \cup A_2)$ and at least one solution is a min-cut, if the $v_i$-$v_{j+n/2}$ min-cut is of size at most $k$.

\textit{Fixing:}
So far, we have only found all pairwise min-cuts, if smaller than $k$.
As a final step, we describe how to extract all latest $\le$$k$-cuts (earliest $\le$$k$-cuts follow by a symmetrical approach).
For any pair $v_i$,$v_j$, let $y_1,y_2,\ldots,y_d$ be all heads of the $v_i$-$v_j$ min-cut found in the previous step.
We first recursively ask for all latest $y_1$-$v_j$, $y_2$-$v_j$, \ldots, $y_d$-$v_j$ $\le$$k$-cuts, for the particular pairs that were not computed already.
Then, we build the corresponding instance of \witness, which  by Lemma~\ref{lem:algoruntime} can be solved by Algorithm~\ref{alg:pruning} in time $2^{\bigO(k^2)} \cdot \textrm{poly}(k) = 2^{\bigO(k^2)}$, and by Theorem~\ref{th:arcsplitcuts2} contains among its solution all desired latest cuts.
A filtering procedure as in the proof of Theorem~\ref{th:dynprog} is applied as a final step.
As we compute the latest cuts exactly once for each pair $i,j$, in total this post-processing takes at most $2^{\bigO(k^2)} n^2$ steps.

Finally, we note that $C^{\otimes 2K}$ codes are of length $\bigO( (\textrm{poly}(k \log n))^{4^k}) = \bigO((k \log n)^{4^{k+o(k)}})$, and Algorithm~\ref{alg:boxes} runs in time $\bigO( (4^k \textrm{poly}(k \log n) )^{\bigO(4^k \cdot k)}) = \bigO((k \log n)^{4^{k+o(k)}})$, giving the desired runtime per matrix entry.
\end{proof}

\begin{corollary}
All latest cuts of size at most $k = o(\log \log n)$ for all pairs of vertices of DAG can be found in $\bigO(n^{\omega+o(1)})$ total time.
\end{corollary}

\subparagraph{Acknowledgments.}
We thank Paweł Gawrychowski, Mohsen Ghaffari, Atri Rudra and Peter Widmayer for the valuable discussions on this problem.

{
\ifprocs
\else
\bibliographystyle{alphaurlinit}
\bibliography{references}
}

\end{document}